\documentclass[journal]{IEEEtran}
\usepackage{amsthm, amssymb, amsmath}
\usepackage{verbatim, float}
\usepackage{mathrsfs}
\usepackage{graphics}
\usepackage[usenames,dvipsnames]{pstricks}
\usepackage{epsfig}
\usepackage{pst-grad} 
\usepackage{pst-plot} 
\usepackage{cases}
\usepackage{algorithmic}
\usepackage{bm}

\usepackage[top=1.4cm, bottom=1.4cm, left=1.33cm, right=1.33cm]{geometry}
\usepackage{url}

\usepackage{hhline}

\usepackage{etoolbox}
\makeatletter
\patchcmd{\@makecaption}
  {\scshape}
  {}
  {}
  {}
\makeatother

\usepackage[singlelinecheck=false]{caption}
\usepackage{diagbox}

\theoremstyle{plain}
\newtheorem{theorem}{Theorem}
\newtheorem{lemma}{Lemma}

\newtheorem{corollary}{Corollary}

\theoremstyle{definition}
\newtheorem{definition}{Definition}
\newtheorem{example}{Example}
\newtheorem{remark}{Remark}

\newcommand{\C}{{\mathcal C}}

\renewcommand{\P}{{\mathcal P}}

\newcommand{\ba}{{\boldsymbol a}}
\newcommand{\bb}{{\boldsymbol b}}

\newcommand{\bc}{{\boldsymbol c}}

\newcommand{\ff}{\mathbb{F}}
\newcommand{\ft}{\mathbb{F}_2}

\newcommand{\al}{\alpha}

\newcommand{\supp}{{\sf supp}}

\newcommand{\dist}{{\mathsf{d}}}
\newcommand{\weight}{{\mathsf{wt}}}
\newcommand{\tr}{\mathsf{Tr}}
\newcommand{\spn}{\mathsf{span}_B}
\newcommand{\rank}{\mathsf{rank}_B}
\newcommand{\lcm}{{\sf lcm}}

\newcommand{\define}{\stackrel{\mbox{\tiny $\triangle$}}{=}}

\newcommand{\nin}{\noindent}

\renewcommand{\ge}{\geqslant}

\newcommand{\et}{{\emph{et al.}}}

\newcommand{\Cd}{\mathcal{C}^\perp}
\newcommand{\rsk}{\text{RS}(A,k)}
\newcommand{\grskl}{\text{GRS}(A,k,\boldsymbol{\lambda})}
\newcommand{\grsnkl}{\text{GRS}(A,n-k,\boldsymbol{\lambda})}

\newcommand{\fa}{f(\alpha)}
\newcommand{\fas}{f(\alpha^*)}
\newcommand{\fab}{f(\overline{\alpha})}
\newcommand{\fap}{f(\alpha')}

\renewcommand{\psi}{p^*_i}
\newcommand{\psix}{p^*_i(x)}
\newcommand{\psia}{p^*_i(\alpha)}
\newcommand{\psias}{p^*_i(\alpha^*)}

\newcommand{\psta}{p^*_t(\alpha)}
\newcommand{\pstas}{p^*_t(\alpha^*)}
\newcommand{\pstab}{p^*_t(\overline{\alpha})}
\newcommand{\psoas}{p^*_1(\alpha^*)}

\newcommand{\psoab}{p^*_1(\overline{\alpha})}
\newcommand{\pstwab}{p^*_2(\overline{\alpha})}
\newcommand{\pstmoab}{p^*_{t-1}(\overline{\alpha})}

\newcommand{\pa}{p(\alpha)}
\newcommand{\pas}{p(\alpha^*)}

\newcommand{\pia}{p_i(\alpha)}
\newcommand{\pias}{p_i(\alpha^*)}
\newcommand{\piab}{p_i(\overline{\alpha})}
\newcommand{\piap}{p_i(\alpha')}
\newcommand{\qia}{q_i(\alpha)}
\newcommand{\qias}{q_i(\alpha^*)}
\newcommand{\qiab}{q_i(\overline{\alpha})}
\newcommand{\qiap}{q_i(\alpha')}
\newcommand{\ria}{r_i(\alpha)}
\newcommand{\rias}{r_i(\alpha^*)}
\newcommand{\riab}{r_i(\overline{\alpha})}
\newcommand{\riap}{r_i(\alpha')}

\newcommand{\po}{p_1}

\newcommand{\poas}{p_1(\alpha^*)}

\newcommand{\qo}{q_1}

\newcommand{\qoab}{q_1(\overline{\alpha})}

\newcommand{\roap}{r_1(\alpha')}

\newcommand{\ptmo}{p_{t-1}}
\newcommand{\ptmoa}{p_{t-1}(\alpha)}
\newcommand{\ptmoas}{p_{t-1}(\alpha^*)}
\newcommand{\ptmoab}{p_{t-1}(\overline{\alpha})}
\newcommand{\ptmoap}{p_{t-1}(\alpha')}
\newcommand{\qtmo}{q_{t-1}}

\newcommand{\qtmoab}{q_{t-1}(\overline{\alpha})}

\newcommand{\rtmo}{r_{t-1}}
\newcommand{\rtmoa}{r_{t-1}(\alpha)}
\newcommand{\rtmoas}{r_{t-1}(\alpha^*)}
\newcommand{\rtmoab}{r_{t-1}(\overline{\alpha})}
\newcommand{\rtmoap}{r_{t-1}(\alpha')}

\newcommand{\pt}{p_t}
\newcommand{\pta}{p_t(\alpha)}
\newcommand{\ptas}{p_t(\alpha^*)}
\newcommand{\ptab}{p_t(\overline{\alpha})}
\newcommand{\ptap}{p_t(\alpha')}
\newcommand{\qt}{q_t}
\newcommand{\qta}{q_t(\alpha)}
\newcommand{\qtas}{q_t(\alpha^*)}
\newcommand{\qtab}{q_t(\overline{\alpha})}
\newcommand{\qtap}{q_t(\alpha')}
\newcommand{\rt}{r_t}

\newcommand{\rtap}{r_t(\alpha')}

\newcommand{\qsb}{Q_{\alpha^*,\overline{\alpha}}}

\newcommand{\qps}{Q_{\alpha', \alpha^*}}

\newcommand{\qsbz}{Q_{\alpha^*,\overline{\alpha}}(z)}

\newcommand{\qbpz}{Q_{\overline{\alpha}, \alpha'}(z)}

\newcommand{\qpsz}{Q_{\alpha', \alpha^*}(z)}

\newcommand{\Qab}{Q_{\alpha, \beta}}
\newcommand{\Qabz}{Q_{\alpha, \beta}(z)}
\newcommand{\Qba}{Q_{\beta,\alpha}}
\newcommand{\Qbaz}{Q_{\beta,\alpha}(z)}

\newcommand{\Qbgz}{Q_{\beta, \gamma}(z)}

\newcommand{\Qgaz}{Q_{\gamma,\alpha}(z)}

\newcommand{\ps}{p_{s}}

\newcommand{\psas}{p_{s}(\alpha^*)}

\newcommand{\qs}{q_{s}}

\newcommand{\qsab}{q_{s}(\overline{\alpha})}

\newcommand{\rs}{r_{s}}

\newcommand{\rsap}{r_{s}(\alpha')}

\newcommand{\us}{u_{s}}

\newcommand{\vs}{v_{s}}

\newcommand{\ws}{w_{s}}

\newcommand{\utmt}{u_{t-2}}
\newcommand{\utmo}{u_{t-1}}
\newcommand{\ut}{u_t}
\newcommand{\vtmt}{v_{t-2}}
\newcommand{\vtmo}{v_{t-1}}
\newcommand{\vt}{v_t}
\newcommand{\wtmt}{w_{t-2}}
\newcommand{\wtmo}{w_{t-1}}
\newcommand{\wt}{w_t}

\newcommand{\Kab}{K_{\alpha,\beta}}
\newcommand{\Kba}{K_{\beta,\alpha}}

\newcommand{\Kbg}{K_{\beta,\gamma}}

\newcommand{\Kabg}{K_{\alpha,\beta,\gamma}}

\newcommand{\Ksb}{K_{\alpha^*,\overline{\alpha}}}
\newcommand{\Kbp}{K_{\overline{\alpha},\alpha'}}
\newcommand{\Kps}{K_{\alpha',\alpha^*}}
\newcommand{\Ksbp}{K_{\alpha^*,\overline{\alpha},\alpha'}}

\newcommand{\cF}{{\sf{char}}(F)}
\newcommand{\is}{i^*}
\newcommand{\ib}{\overline{i}}

\newcommand{\ca}{c_{\alpha}}
\newcommand{\cas}{c_{\alpha^*}}
\newcommand{\cab}{c_{\overline{\alpha}}}
\newcommand{\as}{\alpha^*}
\newcommand{\ab}{\overline{\alpha}}

\begin{document}

\title{Repairing Reed-Solomon Codes\\ With Multiple Erasures
\thanks{
Part of this work~\cite{DauDuursmaKiahMilenkovicTwoErasures2017} was presented at the IEEE International Symposium on Information Theory at Aachen, Germany, in June, 2017.
}
\thanks{
H. Dau is with the Department of Electrical and Computer Systems Engineering, Monash University. Email: hoang.dau@monash.edu. This work was done when he was with the University of Illinois at Urbana-Champaign. O. Milenkovic is with the Coordinated Science Laboratory, University of Illinois at Urbana-Champaign, 1308 W. Main Street, Urbana, IL 61801, USA. Email: milenkov@illinois.edu.
I. Duursma is with the Department of Mathematics, and also with the Coordinated Science Laboratory, University of Illinois at Urbana-Champaign, 1409 W. Green St, Urbana, IL 61801, USA. Email: duursma@illinois.edu.
H. M. Kiah is with the Division of Mathematical Sciences, School of Physical and Mathematical Sciences, Nanyang Technological University, 21 Nanyang Link, Singapore 637371. Email: hmkiah@ntu.edu.sg.} 
}
\author{Hoang Dau, \emph{Member}, \emph{IEEE}, Iwan Duursma, Han Mao Kiah, \emph{Member}, \emph{IEEE}, and Olgica Milenkovic, \emph{Fellow}, \emph{IEEE}}

\date{}
\maketitle

\begin{abstract}
Despite their exceptional error-correcting properties, Reed-Solomon codes have been overlooked in distributed storage applications 
due to the common belief that they have poor repair bandwidth: A naive repair approach would require the whole file to be reconstructed in 
order to recover a single erased codeword symbol. In a recent work, Guruswami and Wootters (STOC'16) proposed a single-erasure 
repair method for Reed-Solomon codes that achieves the optimal repair bandwidth amongst all linear encoding schemes. 
Their key idea is to recover the erased symbol by collecting a sufficiently large number of its traces, each of which can be constructed from 
a number of traces of other symbols. 
We extend the trace collection technique to cope with two and three erasures.     
\end{abstract}

\section{Introduction}
\label{sec:intro}

\subsection{Background}

The \emph{repair bandwidth} is an important performance metric of erasure codes in the context of distributed storage~\cite{Dimakis_etal2007, Dimakis_etal2010}.
In such a system, for a chosen field $F$, a data vector in $F^k$ is mapped to a codeword vector in $F^n$, whose entries are 
stored at different storage nodes. When a node fails, the symbol stored at that node is erased (lost). 
A replacement node (RN) has to recover the content
stored at the failed node by downloading repair data from the other nodes. 
The repair bandwidth is the total amount of information that the RN has to download in order to successfully complete the repair process. 
Bandwidth-efficient repair schemes not only allow low-cost recovery of permanently failed nodes, but also improve the performance of \emph{degraded reads} during transient node failures.

Reed-Solomon (RS) codes~\cite{ReedSolomon1960}, which have been extensively studied in theory~\cite{MW_S} and widely used in practice, were believed to have prohibitively high repair bandwidth. In a naive repair scheme, recovering the content stored at a \emph{single} failed node would require downloading the \emph{whole} file, i.e., $k$ symbols over $F$. The poor performance in repairing failed nodes of RS codes motivated the introduction of repair-efficient codes such as regenerating codes~\cite{Dimakis_etal2007, Dimakis_etal2010, Dimakis_etal2010_survey} and locally repairable codes~\cite{OggierDatta2011,GopalanHuangSimitciYekhanin2012,PapailiopoulosDimakis2012}.

\begin{figure}[t]
\centering
\includegraphics[scale=1]{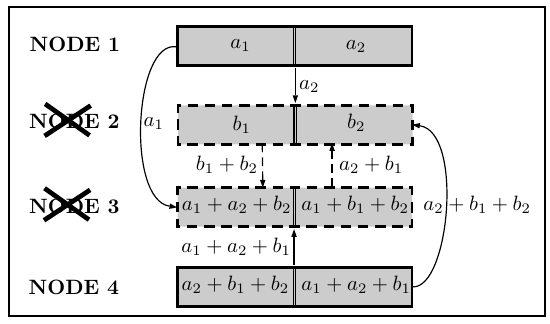}
\caption{A toy example illustrating the repair procedure for \emph{two} failed nodes in a four-node storage system
based on a $[4,2]$ Reed-Solomon code over $\ff_4$. 
The stored file is $\big((a_1,a_2),(b_1,b_2)\big) \in \ff_4^2$, where $a_1$, $a_2$, $b_1$, and $b_2$ are bits in $\ff_2$. 
Suppose that Node~2 and Node~3 fail
simultaneously. In the Download Phase, each replacement node downloads two bits (along the solid arrows) from the two available nodes, namely Node~1 and Node~4. In the Collaboration Phase, the replacement nodes communicate with each other to complete their own repair processes by exchanging two extra bits (along the dashed arrows), computed based on the previously downloaded bits. 
}
\label{fig:toy_example}
\end{figure}

Guruswami and Wootters~\cite{GuruswamiWootters2016} recently proposed a bandwidth-optimal linear repair method based 
on RS codes. The key idea behind their method is to recover a single erased symbol by collecting a sufficiently large number of
its (field) traces, each of which can be constructed from a number of traces of other 
symbols. As all traces belong to a subfield $B$ of $F$ and traces from the same symbol are related, the total repair bandwidth can be significantly reduced.
The repair scheme obtained by Guruswami and Wootters~\cite{GuruswamiWootters2016}, however, only applies to the case
of one erasure, or in other words, one failed node. 

\subsection{Our Contribution}

We propose an extension of the Guruswami-Wootters repair scheme that can ensure recovery from multiple erasures, more precisely two or three erasures. In particular, we provide two \emph{distributed} repair schemes for RS codes for the case of \emph{two} erasures, both of which use
the same repair bandwidth ($n-1$ sub-symbols) \emph{per erased symbol}, equal to that of a single erasure (although note that the lower bound for the repair bandwidth per erasure in the case of multiple erasures may not be the same as that for the case of a single erasure).
In these schemes, the two RNs first download repair data
from all available nodes (Download Phase). They subsequently collaborate to exchange the data in order to complete the repair process at each node (Collaboration Phase).
The first repair scheme requires one round of collaboration, that is, in the Collaboration Phase, the two RNs send out repair data to each other simultaneously in one round. This scheme works whenever the field extension degree $t$ is divisible by the characteristic of the field $F$. An example illustrating the first scheme is given in Fig.~\ref{fig:toy_example}. The second scheme, however, requires two rounds of collaboration, that is, in the Collaboration Phase, one RN receives the repair data from the other RN, completes its repair process, and then sends out
its repair data to the other node. This scheme applies to \emph{all} field extension degrees. Both distributed repair schemes can be easily modified to obtain \emph{centralized} repair schemes, where the two erased symbols are recovered simultaneously by a repair center. 
 
We also develop a \emph{centralized} repair scheme to cater to \emph{three} erasures, and precisely characterize the patterns of three erasures that can be
recovered at a low cost. 
This scheme can be slightly modified to produce a \emph{distributed} repair scheme with a repair bandwidth of $n-1$ sub-symbols from $B$ per erased symbol, which equals the bandwidth needed for a single erasure repair. This distributed scheme requires three rounds of collaboration. Both schemes work whenever the extension degree $t$ is divisible by the characteristic of $F$.

\begin{table}[htb]
\tabcolsep=0.1cm
	\centering
		\begin{tabular}{|l|l|l|}
		\hline
			\textbf{Storage system} & \textbf{Erasure code} & $\#$ \textbf{tolerable failures}\\
			\hhline{|=|=|=|}
			Google File System II (Colossus)~\cite{Fikes2010} & RS(9,6) & Three\\
			\hline
			Quantcast File System~\cite{QFS2013} & RS(9,6) & Three\\
			\hline
			Hadoop Distributed File System~\cite{HDFS-EC16} & RS(9,6) & Three\\
			\hline
			Yahoo Cloud Object Store~\cite{YahooObjectStore} & RS(11,8) & Three\\
			\hline
			Backblaze Vaults~\cite{Backblaze2015} & RS(20,17) & Three\\
			\hline
			Facebook f4 Storage System~\cite{FBf42014} & RS(14,10) & Four\\
			\hline
			Baidu Atlas Cloud Storage~\cite{BaiduAtlas2015} & RS(12,8) & Four\\
			\hline
		\end{tabular}
		\caption{Some well known storage systems and their underlying Reed-Solomon codes (more precisely, their underlying \emph{generalized} Reed-Solomon codes (see Section~\ref{subsec:def_not})). However, it is common practice to refer to them as Reed-Solomon codes.}
		\label{tab:popular_systems}
\end{table}
 
It is worth pointing out that most of the current distributed storage systems offer protection
against at most three or four node failures, which mitigates the need for higher-order failure correction (see Table~\ref{tab:popular_systems}).
Multiple failures, although significantly less frequently than
single failures, do occur in these storage systems. For instance, it was observed~\cite{Rashmi2013,Rashmi2014} that in the 
data-warehouse cluster in production at Facebook, a median of $180$ Terabytes of data is transferred through the top-of-racks switches every day to repair $95,500$ blocks of RS-coded data, $98.08\%$ of which are single failures, $1.87\%$ are double failures, and $0.05\%$ are triple or higher order failures. Multiple-failure repair schemes are also useful in storage systems that
adopt a \emph{lazy repair} model such as Total Recall~\cite{Bhagwan-et2004} or liquid storage systems~\cite{Luby_etal2017}.
In such a system, one waits until the number of failures reaches a certain threshold before performing the repair procedure. If designed properly, this repair mode can help to reduce the network recovery traffic, while still guaranteeing an acceptable low risk of permanent system failure~\cite{Giroire2010, Silberstein2014}.  

\subsection{Organization}    

The paper is organized as follows. We first provide relevant definitions and introduce the terminology used throughout the paper. We then proceed to 
discuss the naive repair scheme as well as the Guruswami-Wootters repair scheme for RS codes in the presence of a single erasure in Section~\ref{sec:GW}. Our main results -- repair schemes for RS codes in the presence of two erasures and three erasures -- are presented in Section~\ref{sec:two_erasures} and Section~\ref{sec:three_erasures}. 
We conclude the paper in Section~\ref{sec:conclusions} by discussing some open problems and related work.

\section{Repairing Reed-Solomon Codes With One Erasure}
\label{sec:GW}

We start by introducing relevant definitions and the notation used in all subsequent derivations, and then proceed to 
review the naive repair approach as well as the approach proposed by
Guruswami and Wootters~\cite{GuruswamiWootters2016} for repairing a single erasure/node failure in RS codes. 

\subsection{Definitions and Notations}
\label{subsec:def_not}

Let $[n]$ denote the set $\{1,2,\ldots,n\}$. Let $B = \text{GF}(p^m)$ be the finite field of $p^m$ elements, for some prime $p$
and $m \geq 1$. Let $F = \text{GF}(p^{mt})$ be a field extension of $B$, where $t \geq 1$. The prime number $p$ is called the \emph{characteristic} of $F$, denoted $\cF$.
We ofter refer to the elements of $F$ as \emph{symbols} and the elements of $B$ as \emph{sub-symbols}. 
We can also treat $F$ as a vector space of dimension $t$ over $B$, i.e. $F \cong B^t$, and hence each symbol in $F$ may be represented 
as a vector of length $t$ over $B$. The rank of a set $V \subseteq F$ over $B$ is denoted by $\rank(V)$.

A linear $[n,k]$ code $\C$ over $F$ is a subspace of $F^n$ of dimension $k$. Each element of a code is referred
to as a codeword. 
The dual of a code $\C$, denoted $\Cd$, is the orthogonal complement of $\C$. 
The support of a vector $\bc = (c_1,\ldots,c_n) \in F^n$ is defined as $\supp(\bc) = \{j \colon c_j \neq 0\}$. 
The Hamming weight of a vector $\bc \in F^n$, denoted $\weight(\bc)$, is 
the size of its support $\supp(\bc)$. 
The Hamming distance of any two vectors $\bc$ and $\bc'$ in $F^n$, denoted $\dist(\bc,\bc')$, is defined by $\weight(\bc-\bc')$. The minimum distance of a code $\C$, denoted $\dist(\C)$, is the minimum Hamming distance between any two distinct codewords of $\C$. The Singleton bound establishes that for any $[n,k]$ code of minimum distance $d$, 
the inequality $d \leq n - k + 1$ holds (see~\cite{MW_S}). A code that has minimum distance attaining this bound is called a maximum-distance separable (MDS) code.
 
\begin{definition} 
\label{def:RS}
Let $F[x]$ denote the ring of polynomials over $F$. The Reed-Solomon code
$\text{RS}(A,k) \subseteq F^n$ of dimension $k$ over a finite
field $F$ with evaluation points $A=\{\alpha_1,\alpha_2,\ldots, \alpha_n\}\subseteq F$
is defined as: 
\[
\text{RS}(A,k) = \Big\{\big(f(\alpha_1),\ldots,f(\alpha_n)\big) \colon f \in F[x],\ \deg(f) < k \Big\}.
\]
\end{definition} 
A \emph{generalized} Reed-Solomon code, $\grskl$, where $\boldsymbol{\lambda} = (\lambda_1,\ldots,\lambda_n)\in F^n$, is defined similarly to a Reed-Solomon code, except that the codeword
corresponding to a polynomial $f$ is now defined as $\big( \lambda_1f(\alpha_1),\ldots,\lambda_n f(\alpha_n) \big)$, $\lambda_i \neq 0$ for all $i \in [n]$.
It is well known that the dual of an RS code $\rsk$, for any $n \leq |F|$, is a generalized RS code $\grsnkl$, for some multiplier vector $\boldsymbol{\lambda}$~(see~\cite[Chp.~10]{MW_S}). Whenever clear from the context, we use $f(x)$ to denote a polynomial of degree at most $k-1$, which corresponds to a codeword of the RS code $\C=\rsk$, and $p(x)$ to denote a polynomial of degree at most $n-k-1$, which corresponds to a dual codeword in $\Cd$. Since
$
\sum_{\alpha \in A}\pa(\lambda_{\alpha}\fa) = 0, 
$
we refer to such a polynomial $p(x)$ as a check polynomial for $\C$. 
Note that when $n = |F|$, we have $\lambda_\alpha = 1$ for all $\alpha \in F$. 
In general, as recovering $\fa$ is equivalent to recovering $\lambda_{\alpha}\fa$, to simplify the notation, we often omit the factor $\lambda_\alpha$ in the equation above.


\subsection{Naive Repair of One Erasure}

Suppose that the polynomial $f(x) \in F[x]$ corresponds to a codeword in the RS code $\C=\text{RS}(A,k)$ and that $f(\alpha^*)$ is the erased 
symbol, where $\alpha^* \in A$ is an evaluation point of the code. 
Pick a check $p(x) \in F[x]$ of $\C$ such that $p(\alpha^*) \neq 0$ and the Hamming weight of the dual codeword corresponding to $p(x)$ is precisely equal to $k+1$. 
The existence of such check is guaranteed by the fact that in an MDS code of length $n$ and minimum distance $n-(n-k)+1=k+1$, 
each subset of $[n]$ of size $k+1$ is always the support of some codeword. The check $p(x)$ generates the following 
\emph{repair equation}:
\begin{equation} 
\label{eq:naive_repair}
\pas \fas = -\sum_{\alpha \in A \setminus \{\alpha^*\}}\pa\fa. 
\end{equation}   
Since there are precisely $k$ evaluation points $\alpha \neq \alpha^*$ where $p(\alpha)\neq 0$, \eqref{eq:naive_repair}
implies that $\fas$ can be computed from $\fa$ for $k$ different $\alpha \in A\setminus \{\alpha^*\}$. 
Hence, the naive repair scheme requires downloading $k$ symbols over $F$, or equivalently, $kt$ sub-symbols over $B$, 
which is equal to the size of the stored file. 
Although this is a well-known fact regarding RS codes, it provides a smooth transition from the naive repair to the trace repair framework, replacing the one repair equation~\eqref{eq:naive_repair} with $t$ equations~\eqref{eq:GW_repair}.

\subsection{The Guruswami-Wootters Repair Scheme for One Erasure}
\label{subsec:GW}

Given that $F$ is a field extension of $B$ of degree $t$, i.e. $F = \text{GF}(p^{mt})$ and $B = \text{GF}(p^m)$, for 
some prime $p$, one may define the field trace of any symbol $\alpha \in F$ as
\[
\mathsf{Tr}_{F/B}(\alpha) = \sum_{i = 0}^{t-1} \alpha^{|B|^i},
\]
which is always a sub-symbol in $B$. We often omit the subscript $F/B$ for succinctness. 
The key points in the repair scheme proposed by Guruswami and Wootters~\cite{GuruswamiWootters2016} can be 
summarized as follows:
\begin{itemize}
	\item Each symbol in $F$ can be recovered from its $t$ independent traces, which we call the \emph{target} traces. More precisely, given a 
	basis $u_1,u_2,\ldots,u_t$ of $F$ over $B$, any $\alpha \in F$ can be uniquely determined given the values of $\tr(u_i\,\alpha)$ for $i\in [t]$, i.e. $\alpha = \sum_{i=1}^t\tr(u_i \alpha)u^\perp_i$, where $\{u^\perp_i\}_{i=1}^t$ is the dual (trace-orthogonal) basis of $\{u_i\}_{i=1}^t$ (see, for instance~\cite[Ch.~2, Def.~2.30]{LidlNiederreiter1986}).
	\item When $n-k\ge |B|^{t-1}$, the $t$ target traces can be computed from $n-1$ sub-symbols, referred to as the \emph{repair} traces, downloaded from the $n-1$ available nodes. Thus, the repair bandwidth is $n-1$ sub-symbols.
\end{itemize}
Note that the checks of $\C$ are precisely those polynomials $p(x) \in F[x]$ where $\deg(p) < n - k$. 
It turns out that when $n-k\geq |B|^{t-1}$, we can define checks that take part in the repair process via the trace function. More specifically, for each $u \in F$ and $\alpha \in F$, we define the polynomial
\begin{equation} 
\label{eq:p}
p_{u,\alpha}(x) = \tr\big(u(x-\alpha)\big)/(x-\alpha).
\end{equation} 
By the definition of the trace function, the following lemma follows in a straightforward manner. 

\begin{lemma}[\cite{GuruswamiWootters2016}]
\label{lem:p}
The polynomial $p_{u,\alpha}(x)$ defined in~\eqref{eq:p} satisfies the following properties.\\
\quad (a) $\deg(p_{u,\alpha}) = |B|^{t-1}-1$;\quad \text{(b)} $p_{u,\alpha}(\alpha) = u$.
\end{lemma}

By Lemma~\ref{lem:p}~(a), $\deg(p_{u,\alpha}) = |B|^{t-1}-1 < n - k$. Therefore, the polynomial $p_{u,\alpha}(x)$ corresponds to a codeword of $\Cd$ and is a check for $\C$. 
Now let $U = \{u_1,\ldots,u_t\}$ be a basis of $F$ over $B$, and set
\[
p_i(x) \define p_{u_i,\alpha^*}(x) = \tr\big(u_i(x-\alpha^*)\big)/(x-\alpha^*), \quad i \in [t].
\]  
These $t$ polynomials correspond to $t$ codewords of $\Cd$. 
Therefore, we obtain $t$ equations of the form 
\begin{equation} 
\label{eq:GW_repair}
\pias\fas = - \sum_{\alpha \in A \setminus \{\alpha^*\}} \pia\fa, \quad i \in [t]. 
\end{equation} 
An essential step in the Guruswami-Wootters repair scheme is to apply the trace function to
both sides of \eqref{eq:GW_repair} to obtain $t$ different \emph{repair equations}
\begin{equation} 
\label{eq:GW_repair_trace}
\tr\big(\pias\fas\big) = - \sum_{\alpha \in A \setminus \{\alpha^*\}} \tr\big(\pia\fa\big), \ i \in [t]. 
\end{equation}
According to Lemma~\ref{lem:p}~(b), $\pias = u_i$, for $i = 1,\ldots,t$. Moreover, by the
linearity of the trace function, we also have
\[
\tr\big(\pia\fa\big) = \tr\big(u_i(\alpha -\alpha^*)\big) \times \tr\Big(\dfrac{f(\alpha)}{\alpha-\alpha^*}\Big). 
\]
Therefore, one can rewrite~\eqref{eq:GW_repair_trace} as follows: 
\begin{equation}
\label{eq:GW_repair_trace_explicit}
\tr\big(u_i\fas\big) = - \hspace{-10pt} \sum_{\alpha \in A \setminus \{\alpha^*\}}\hspace{-10pt} \tr\big(u_i(\alpha -\alpha^*)\big) \times \tr\Big(\dfrac{f(\alpha)}{\alpha-\alpha^*}\Big), \ i \in [t]. 
\end{equation}
The right-hand side sums of the equations~\eqref{eq:GW_repair_trace_explicit} can be computed by downloading the \emph{repair} trace 
$\tr\Big(\frac{f(\alpha)}{\alpha-\alpha^*}\Big)$ from the node storing $\fa$, for each
$\alpha \in A \setminus \{\alpha^*\}$.
As a consequence, the $t$ independent \emph{target} traces $\tr\big(u_i\fas\big)$, $i = 1,\ldots,t$, of $\fas$ can be determined by downloading one sub-symbol from each of the $n-1$ available nodes. The erased symbol $\fas$ can subsequently be recovered from its $t$ independent target traces. 
By~\cite[Cor.~1]{DauMilenkovic2017}, this scheme is bandwidth-optimal when $n = |F|$ and $k = n(1-1/|B|)$.


\section{Repairing Reed-Solomon Codes with Two Erasures}
\label{sec:two_erasures}

We consider the same setting as in Section~\ref{subsec:GW}, i.e. $n - k \geq |B|^{t-1}$, where $B = \text{GF}(p^m)$ and $F = \text{GF}(p^{mt})$, 
and assume as before that $\C$ is a Reed-Solomon code $\rsk$ over $F$ (see Definition~\ref{def:RS}). 
However, we now suppose that two storage nodes failed, or equivalently, that two codeword symbols, say $\fas$ and $\fab$, 
are erased. Two distributed repair schemes are proposed, both of which use the same bandwidth per erasure as for the case of a single erasure~\cite{GuruswamiWootters2016}.

A typical distributed repair scheme consists of two phases: The Download Phase and the Collaboration Phase. In the Download Phase, the replacement nodes contact and download recovery data from all available nodes. In the Collaboration Phase, the replacement nodes collaborate by exchanging information among themselves to help each other complete their repair processes. The repair bandwidth (per erasure) is defined as the amount of data each replacement node must download from the other nodes during its repair process. 

\subsection{General Idea}
\label{subsec:general_idea}

We first discuss the challenges associated with repairing two erased symbols and then present some basic results that are useful for both Section~\ref{subsec:depth_one} and Section~\ref{subsec:depth_two}.

A check $p(x)$ is said to \emph{involve} a codeword symbol $\fa$ if $\pa \neq 0$. Otherwise, $p(x)$ is said to \emph{exclude} $\fa$.  When only one symbol $\fas$ is erased, every check $p(x)$ that involves $\fas$ can be used to generate a repair equation as follows.
\begin{equation} 
\label{eq:trace_repair}
\tr\big(\pas \fas\big) = -\sum_{\alpha \in A \setminus \{\alpha^*\}}\tr\big(\pa\fa\big). 
\end{equation}

However, when two symbols $\fas$ and $\fab$ are erased, in order to, say, recover $\fas$, we no longer have
the freedom to use every possible check that involves $\fas$. Indeed, those checks that involve both $\fas$ and $\fab$
cannot be used in a straightforward manner for repair, because we cannot simply compute the right-hand side sum of \eqref{eq:trace_repair} without retrieving some information from $\fab$.

The gist of our approach for the distributed repair schemes is to first generate those checks that only involve one erased symbol, $\fas$ or $\fab$, but not both. 
We show that there exist $2(t - 1)$ such checks, which are used in the Download
Phase. Each RN uses $t-1$
checks and downloads the corresponding $n-2$ sub-symbols from each available node. 
Apart from the $t-1$ checks that involve $\fas$ but not $\fab$, and the $t-1$ checks that involve $\fab$ but not $\fas$, we also introduce
\emph{two} additional checks that involve both $\fas$ and $\fab$, which are useful
in the Collaboration Phase. 
It is not immediately clear how these last two checks can be used at all. 
However, we prove that when the extension degree $t$ is divisible by the characteristic of the field $F$, each erased symbol can be recovered at each
RN using the aforementioned $t$ checks, at the cost of downloading in
total $n-1$ sub-symbols from $n-2$ surviving nodes and from the other RN. In the first repair scheme, the two RNs exchange their repair data simultaneously (parallel repair), while in the second scheme, one node \emph{waits} to receive the data from the other node before sending out its own repair data (sequential repair). By allowing one node to wait in the Collaboration Phase, we obtain a repair scheme that works for every field extension degree. Each distributed scheme can be easily modified to yield a centralized scheme by removing the Collaboration Phase. 

To identify check equations that involve one codeword symbol $f(\alpha)$ 
but not the other symbol $f(\beta)$, we first introduce a special polynomial $\Qabz$, defined as follows:
\begin{equation}
\label{eq:Q}
\Qabz = \tr\big(z(\beta - \alpha)\big), \quad \alpha \neq \beta.
\end{equation} 
Let $\Kab$ denote the root space of $\Qabz$.
Then
\begin{equation} 
\label{eq:Kab}
\Kab = \left\{z^* \in F \colon \tr(z^*\alpha) = \tr(z^*\beta)\right\}.
\end{equation}

\begin{lemma} 
\label{lem:Kab}
The following statements hold for every $\alpha$ and $\beta$ in $F$, $\alpha \neq \beta$.  
\begin{enumerate}
	\item[(a)] $\Kab \equiv \Kba$. In other words, the polynomial $\Qab$ and 
	the polynomial $\Qba$ have the same root spaces.
	\item[(b)] $\dim_B(\Kab) = \dim_B(\Kba) = t - 1$. 
\end{enumerate}
\end{lemma} 
\begin{proof} 
From \eqref{eq:Kab}, due to symmetry, $\Kab \equiv \Kba$.
As the trace function is a linear mapping from $F$ to $B$, its kernel $K = \{\kappa \in F \colon \tr(\kappa)=0\}$ is a $B$-subspace of dimension $t-1$ (see, for instance~\cite[Thm.~2.23]{LidlNiederreiter1986}). Hence, the root space of $\Qabz$, i.e. $\Kab = \frac{1}{\beta-\alpha}K$, is also a $B$-subspace of dimension $t-1$. 
\end{proof}

We then use a root $z^*$ of the polynomial $\Qabz$ to define a check equation according to \eqref{eq:p}.
\[
p_{z^*,\alpha}(x) = \tr\big(z^*(x - \alpha) \big)/(x - \alpha). 
\]
The following properties of $p_{z^*,\alpha}(x)$ will be used in our subsequent proofs.

\begin{lemma} 
\label{lem:Qp}
Suppose that $\alpha$ and $\beta$ are two distinct elements of $F$,
and $z^*$ is a root of $\Qabz$ or $\Qbaz$ in $F$, i.e. $z^* \in \Kab$. 
Then the following claim holds.
\begin{enumerate}
	\item[(a)] $p_{z^*,\alpha}(\beta) = 0$. 
\end{enumerate}
Moreover, if the extension degree $t$ is divisible by $\cF$ then 
\begin{enumerate}
	\item[(b)] $p_{u,\alpha}(\beta)$ is a root of $\Qabz$ and $\Qbaz$, for every $u \in F$, and so is $1/(\beta-\alpha)$.
\end{enumerate}
\end{lemma}
\begin{proof} 
Note that according to Lemma~\ref{lem:Kab}~(a), the root spaces of $\Qabz$ and $\Qbaz$
are the same. The first claim is clear based on the definitions of $\Qabz$ and $p_{z^*,\alpha}(x)$.
For the second claim, it is sufficient to show that both $1/(\beta-\alpha)$ and $p_{u,\alpha}(\beta)$ is a root of $\Qabz$.
By the definition of $\Qabz$, we have 
$\Qabz = \tr\big(z(\beta - \alpha) \big)$.
Therefore, \vspace{-10pt}
\[
\Qab\big(1/(\beta-\alpha)\big) = \tr(1) = \sum_{i = 0}^{t-1} 1 = t = 0, 
\]
whenever $t$ is divisible by the characteristic of the field.
Hence $1/(\beta-\alpha)$ is a root of $\Qabz$. Since $p_{u,\alpha}(\beta) = b/(\beta - \alpha)$, where $b = \tr\big(u(\beta - \alpha)\big) \in B$, we deduce that 
\[
\Qab(p_{u,\alpha}(\beta)) = b \Qab\big(1/(\beta-\alpha)\big) = 0,
\] 
because of the linearity of $\Qabz$. Thus, $p_{u,\alpha}(\beta)$ is also a root of $\Qabz$.  
\end{proof} 

The following lemma restates what is shown in Section~\ref{subsec:GW}.  

\begin{lemma} 
\label{lem:trace}
For $\alpha \neq \alpha^*$ and $u \in F$, \vspace{-5pt}
\[
\tr\big( p_{u,\alpha^*}(\alpha)\fa\big) 
= \tr\big( u(\alpha-\alpha^*) \big) \times \tr\Big(\dfrac{\fa}{\alpha - \alpha^*} \Big).
\]  
Hence, the trace $\tr\big( p_{u,\alpha^*}(\alpha)\fa\big)$ can be determined by downloading the repair trace $\tr\big(\frac{\fa}{\alpha - \alpha^*} \big)$ from the node storing $\fa$. 
\end{lemma} 

\subsection{A One-Round Distributed Repair Scheme for Two Erasures}
\label{subsec:depth_one}

The scheme comprises of two phases, the Download Phase, where each RN contacts and downloads data from the other $n-2$ available nodes, and the 
Collaboration Phase, where the two RNs exchange the data, based on what
they receive earlier in the Download Phase. The main task is to design
the data to be exchanged during the two phases. This task can be completed via 
a selection of proper check polynomials to be used by each RN. 
We discuss the generation of these polynomials below. 

Let $\Ksb$ be the root space of the polynomial $\qsbz$. By Lemma~\ref{lem:Kab}~(b), $\dim_B(\Ksb) = t-1$. Let $U = \{u_1,u_2,\ldots,u_{t-1}\} \subseteq F$ and $V = \{v_1,v_2,\ldots,v_{t-1}\} \subseteq F$ be two arbitrary bases of $\Ksb$ over $B$.
We extend $U$ and $V$ to obtain the two bases $U' = \{u_1,\ldots,u_t\}$ and $V' = \{v_1,\ldots,v_t\}$ of $F$ over $B$, respectively.
Note that while it is not crucial to choose two different sets $U$ and $V$, it is not compulsory to use the same set $U \equiv V$ either.
For $i \in [t]$, we set 
\begin{equation} 
\label{eq:pix}
p_i(x) \define p_{u_i,\alpha^*}(x) = \tr\big(u_i(x - \alpha^*) \big) / (x - \alpha^*),
\end{equation} 
\begin{equation} 
\label{eq:qix}
q_i(x) \define p_{v_i,\overline{\alpha}}(x) = \tr\big(v_i(x - \overline{\alpha}) \big)/(x - \overline{\alpha}).
\end{equation}
 
\nin\textbf{Download Phase.}
In this phase, each RN contacts $n-2$ available nodes to download repair
data. To determine what to download, the RN for $\fas$ uses the
first $t-1$ checks $p_1,\ldots,p_{t-1}$ to construct the following $t-1$ repair equations. 
\begin{equation}
\label{eq:pi} 
\tr\big(\pias\fas\big) = - \hspace{-10pt} \sum_{\alpha \in A \setminus \{\alpha^*\}} \hspace{-10pt} \tr\big(\pia\fa\big),\ i \in [t-1].
\end{equation} 
Similarly, the RN for $\fab$ creates the following repair equations. 
\begin{equation}
\label{eq:qi}
\tr\big(\qiab\fab\big) = - \hspace{-10pt} \sum_{\alpha \in A \setminus \{\overline{\alpha}\}} \hspace{-10pt} \tr\big(\qia\fa\big),\ i \in [t-1].
\end{equation} 
By Lemma~\ref{lem:Qp}~(a), we have $\piab = 0$ and $\qias = 0$ for
all $i = 1,\ldots,t-1$. Therefore, the right-hand sides of 
\eqref{eq:pi} and \eqref{eq:qi} do not involve $\fas$ and $\fab$. 
As a result, each RN can recover $t-1$ independent target traces of the corresponding
erased symbol by downloading $n-2$ sub-symbols (repair traces) from the available nodes. Corollary~\ref{cr:downloading_phase}, which follows directly from Lemma~\ref{lem:trace}, formally states this
fact. 

\begin{corollary} 
\label{cr:downloading_phase}
In the Download Phase, the replacement node for $\fas$ can recover $t-1$ independent target traces, namely
$\tr\big(\poas\fas \big),\ldots,\tr\big(\ptmoas\fas \big)$, by downloading $n-2$
repair traces, i.e. $\tr\Big(\frac{\fa}{\alpha - \alpha^*} \Big)$ from the available
node storing $\fa$, for all $\alpha \in A \setminus \{\alpha^*, \overline{\alpha}\}$. A similar statement holds for the replacement node for $\fab$, where the checks
are $q_i$ and the repair traces are $\tr\Big(\frac{\fa}{\alpha - \overline{\alpha}} \Big)$. 
\end{corollary} 

\nin\textbf{Collaboration Phase.}
As one more independent target trace of each erased symbol is needed for a complete recovery, the two RNs create two additional repair equations for $\fas$, $\fab$, respectively. 
\begin{equation}
\label{eq:pt}
\tr\big(\ptas\fas\big) + \tr\big(\ptab\fab\big)
= - \hspace{-10pt} \sum_{\alpha \in A \setminus \{\alpha^*,\overline{\alpha}\}} \hspace{-10pt} \tr\big(\pta\fa\big).
\end{equation}
\begin{equation}
\label{eq:qt}
\tr\big(\qtab\fab\big) + \tr\big(\qtas\fas\big)
= - \hspace{-10pt} \sum_{\alpha \in A \setminus \{\alpha^*,\overline{\alpha}\}} \hspace{-10pt} \tr\big(\qta\fa\big).
\end{equation}
It is clear that from the repair traces $\tr\big(\frac{\fa}{\alpha - \alpha^*} \big)$, $\alpha \in A \setminus \{\alpha^*, \overline{\alpha}\}$, retrieved in the
Download Phase, the RHS of~\eqref{eq:pt} can be determined. 
However, to extract the target trace $\tr\big(\ptas\fas\big)$, the RN for $\fas$ needs to know the term $\tr\big(\ptab\fab\big)$, which would have been determined based on the repair trace $\tr\big( \frac{\fab}{\overline{\alpha} - \alpha^*} \big)$ downloaded from the node storing $\fab$ if it had not failed.
The following lemma states that for certain field extension degrees, this missing piece of information can be created by the RN for $\fab$ based on what it obtains in the Download Phase. It can then send this repair trace to the RN for $\fas$ to help complete the recovery of that symbol. A similar scenario also holds for $\fab$.    

\begin{lemma} 
\label{lem:dependence}
If the field expansion degree $t$ is divisible by $\cF$ then both $\ptab$ and $1/(\overline{\alpha} - \alpha^*)$ can be written as $B$-linear combinations of elements in $V = \{\qiab \colon i \in [t-1]\} = \{v_1,\ldots,v_{t-1}\}$. Likewise, 
both $\qtas$ and $1/(\alpha^* - \overline{\alpha})$ can be written as $B$-linear combinations of elements in $U = \{\pias \colon i \in [t-1]\} = \{u_1,\ldots,u_{t-1}\}$. 
\end{lemma} 
\begin{proof} 
Because of symmetry, it suffices to just prove the first statement of the lemma. 
By Lemma~\ref{lem:p}~(b), we have $\qiab = v_i$, for every $i \in [t-1]$. Therefore, 
$\{\qiab \colon i \in [t-1]\} = \{v_1,\ldots,v_{t-1}\} = V$, which is a basis of the root space $\Ksb$ of the polynomial $\qsbz$. 
Therefore, in order to show that $\ptab$ and $1/(\overline{\alpha} - \alpha^*)$ are $B$-linear combinations of elements in $V$, it is sufficient to prove that they are roots of $\qsbz$.
But this follows immediately from Lemma~\ref{lem:Qp}~(b), because $p_t(x)$ equals $p_{u_t,\alpha^*}(x)$ by its definition in \eqref{eq:pix}. 
\end{proof} 

From the linearity of the trace function, we arrive at the following corollary of Lemma~\ref{lem:dependence}. 

\begin{corollary}
\label{cr:dependence}
If the field expansion degree $t$ is divisible by $\cF$
then the trace $\tr\big(\ptab\fab\big)$ and the repair trace $\tr\big(\frac{\fab}{\overline{\alpha} - \alpha^*}\big)$ can be written as $B$-linear combinations of the target traces in
$\left\{\tr\big(\qiab\fab \big) \colon i \in [t-1]\right\}$. Also, the trace $\tr\big(\qtas\fas\big)$ the repair trace $\tr\big(\frac{\fas}{\alpha^*-\overline{\alpha}}\big)$ can be
written as $B$-linear combinations of the target traces in $\left\{\tr\big(\pias\fas\big) \colon i \in [t-1]\right\}$. 
Moreover, the coefficients of these combinations do not depend on $f$. 
\end{corollary} \vspace{-10pt}

\begin{figure}[htb]
\centering
\includegraphics[scale=.9]{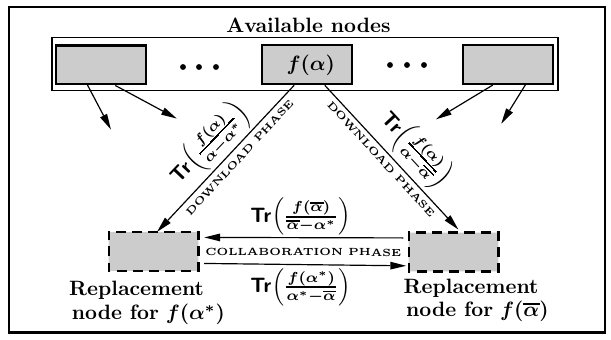}
\caption{Illustration of the \emph{one-round} distributed repair scheme for two erasures in Reed-Solomon codes. In the Collaboration Phase, the two repair
traces are exchanged \emph{simultaneously}.}
\label{fig:distributed_two}
\vspace{-5pt}
\end{figure}

Recall that by Lemma~\ref{lem:trace}, the traces $\tr\big(\ptab\fab\big)$ and $\tr\big(\qtas\fas\big)$ can be determined based on the repair traces $\tr\Big(\frac{\fab}{\overline{\alpha} - \alpha^*} \Big)$
and $\tr\Big(\frac{\fas}{\alpha^* - \overline{\alpha}} \Big)$, respectively. Therefore, in the Collaboration Phase, the RNs can send their repair data to each other, which matches precisely what they would have sent if they had not failed.
The graphical illustration of the two phases of this scheme is depicted in Fig.~\ref{fig:distributed_two}. We refer to this as a \emph{one-round} distributed repair scheme because in the Collaboration Phase, two RNs exchange repair data in one round and do not have to wait for each other.
It is also clear that in this scheme, each RN does \emph{not} need to know which basis/checks the other RN is using. 

\begin{lemma} 
\label{lem:collaboration_phase}
In the Collaboration Phase, the replacement node for $\fas$ can recover its $t$-th target trace $\tr\big(\ptas\fas \big)$, by downloading one repair trace $\tr\Big(\frac{\fab}{\overline{\alpha} - \alpha^*} \Big)$ from the replacement node for $\fab$. Similarly, the replacement node for $\fab$ can recover its $t$-th target trace $\tr\big(\qtab\fab \big)$ by downloading one repair trace $\tr\Big(\frac{\fas}{\alpha^* - \overline{\alpha}} \Big)$ from the replacement node for $\fas$. 
\end{lemma} 

Before stating our main theorem of this section,
we summarize in Table~\ref{tab:two_erasures} the way we use different types of checks for the recovery of two erased symbols $\fas$ and $\fab$.  

\begin{theorem} 
\label{thm:depth_one}
The one-round distributed repair scheme can be used to repair any two erased symbols of a Reed-Solomon codes $\rsk$ at a repair bandwidth of $n-1$ sub-symbols per erased symbol, given that $n-k \geq |B|^{t-1}$ and
$\cF$ divides $t$.
\end{theorem}  
\begin{proof}\noindent 
By Lemma~\ref{lem:p}~(b), $\pias = u_i$ and $\qiab = v_i$ for $i \in [t]$.
Recall that the sets $U' = \{u_1,\ldots,u_t\}$ and $V' = \{v_1,\ldots,v_t\}$ are both  linearly independent over $B$. Therefore, after the two phases, each RN obtains $t$ independent target traces for each erased symbol, $\tr\big(u_i\fas\big)$, for $\fas$, and $\tr\big(v_i\fab\big)$, for $\fab$, for all $i \in [t]$.  
Thus, each erased symbol will have $t$ independent traces for its recovery.
Each RN downloads $n-2$ sub-symbols in the Download Phase and one 
sub-symbol in the Collaboration Phase, according to Corollary~\ref{cr:downloading_phase}
and Lemma~\ref{lem:collaboration_phase}, which sum up to a repair bandwidth of $n-1$ sub-symbols. 
\end{proof} 
\vspace{-5pt}

\begin{table}[htb]
\centering
\begin{tabular}{|l|c|c|l|}\hline
\diagbox[width=10em]{\textbf{Checks}}{\textbf{Erased}\\\textbf{Positions}} & $\alpha^*$ & $\overline{\alpha}$
& \textbf{Purpose}\\ 
\hline
$\po,\ldots,\ptmo$ & $\times$ & $\cdot$ & $\to$ target traces for $\fas$ \\ \hline
$\qo,\ldots,\qtmo$ & $\cdot$ & $\times$ & $\to$ target traces for $\fab$ \\ \hhline{|=|=|=|=|}
$\pt$ & $\times$ & $\times$ & $\to$ target trace for $\fas$ \\ \hline
$\qt$ & $\times$ & $\times$ & $\to$ target trace for $\fab$ \\ \hline
\end{tabular}
\caption{The list of checks used for repairing two erased symbols $f(\alpha^*)$ and $\fab$. A cross ``$\times$''
means a nonzero value, while a dot ``$\cdot$'' means a zero value. 
There are $2t-2$ checks, $t-1$ involving only $\fas$ and $t-1$ involving only $\fab$. The target traces generated by these checks also help to produce the second terms in the LHS of the repair equations \eqref{eq:pt} and \eqref{eq:qt} corresponding to $p_t$ and $q_t$, respectively. The
$t$-th target trace for each erased symbol can then be extracted from~\eqref{eq:pt} and \eqref{eq:qt}.}
\label{tab:two_erasures}
\vspace{-10pt}
\end{table}

\begin{example}
\label{ex:2}
Let $q = 2$, $t = 2$, $n = 4$, and $k = 2$. 
Let $\ff_4 = \{0,1,\xi,\xi^2\}$, where $\xi^2+\xi+1 = 0$. Then $\{1,\xi\}$ is a
basis of $F=\ff_4$ over $B=\ff_2$. Moreover, each element $\ba \in \ff_4$ can be represented
by a pair of bits $(a_1,a_2)$ where $\ba = a_1 + a_2\xi$. Suppose the stored file is
$(\ba,\bb) \in \ff_4^2$. To devise a systematic RS code, 
we associate with each file $(\ba,\bb) \in \ff_4^2$ a polynomial $f(x)=f_{\ba,\bb}(x)\define \ba + (\bb-\ba)x$.
We have $f(0) = a_1 + a_2\xi = \ba$, $f(1) = b_1 + b_2\xi = \bb$,
$f(\xi) = (a_1+a_2+b_2) + (a_1+b_1+b_2)\xi$, and
$f(\xi^2) = (a_2+b_1+b_2) + (a_1+a_2+b_1)\xi$.
The four codeword symbols $f(0)$, $f(1)$, $f(\xi)$, and $f(\xi^2)$ are stored at Node~1,
Node~2, Node~3, and Node~4, respectively, as depicted in 
Fig.~\ref{fig:toy_example}.

\textbf{Download Phase.}
Set $Q_{1,\xi}(z) \define \tr\big( z(\xi-1)\big) = \xi z^2 + \xi^2z$.
We choose two bases $U = V = \{\xi\}$ of the root space of $Q_{1,\xi}(z)$. 
Set $p_1(x) = \tr\big(\xi(x-1)\big) / (x-1) = \xi^2x + 1$, and 
$q_1(x) = \tr\big(\xi(x-\xi)\big) / (x-\xi) = \xi^2x+\xi^2$. 
RN2 (RN for Node~2) downloads two bits from the two available nodes, namely $a_2 = \tr\big(f(0) / (0-1)\big)$ from Node~1 and
$a_2+b_1+b_2 = \tr\big(f(\xi^2) / (\xi^2-1)\big)$ from Node~4. 
It then uses \eqref{eq:pi} to obtain the first trace
$b_1+b_2 = \tr(\xi f(1)) = 1\times a_2 + 1\times(a_2+b_1+b_2)$. 
Similarly, RN3 (RN for Node~3) also downloads 
$a_1 = \tr\big(f(0) / (0-\xi)\big)$ from Node~1 and
$a_1+a_2+b_1 = \tr\big(f(\xi^2) / (\xi^2-\xi)\big)$ from Node~4.
It then recovers $a_2+b_1 = \tr(\xi f(\xi)) = 1\times a_1 + 1\times (a_1+a_2+b_1)$.

\textbf{Collaboration Phase.} 
RN2 sends $b_1+b_2$ over to the RN3, which, by Lemma~\ref{lem:dependence}, is the same as $\tr\big(f(1)/(1-\xi)\big)$. Conversely, RN3 sends $a_2+b_1$ over to
RN2, which is the same as $\tr\big(f(\xi)/(\xi-1)\big)$. 
$U$ and $V$ are extended to the basis $\{\xi,\xi^2\}$ of $\ff_4$ over $\ff_2$.
Set $p_2(x) = \tr\big(\xi^2(x-1)\big) / (x-1) = \xi x + 1$, and 
$q_2(x) = \tr\big(\xi^2(x-\xi)\big) / (x-\xi) = \xi x$.
Now, RN2 has three repair traces to recover
the second trace of $f(1)$ using $p_2$, i.e. $b_1 = \tr\big(\xi^2f(1)\big) = 
1\times a_2 + 0\times(a_2+b_1+b_2) + 1\times (a_2 + b_1)$.
Based on the two traces $b_1+b_2$ and $b_1$, the erased symbol $\bb=f(1)$ can be recovered. 
Similarly, RN3 can recover the second trace of $f(\xi)$
as $a_1+a_2+b_2 = \tr\big(\xi^2f(\xi)\big) = 0\times(a_1) + 1\times(b_1+b_2) + 1\times(a_1+a_2+b_1)$, and then can recover $f(\xi)$ completely. 
\end{example}    

\subsection{A Two-Round Distributed Repair Scheme for Two Erasures}
\label{subsec:depth_two}

We modify the one-round repair scheme developed in the previous subsection to obtain
a two-round scheme that works for \emph{all} field extension degrees. 
We still generate the checks $p_1,\ldots,p_t$ and $q_1,\ldots,q_t$ as in the first scheme,
given by \eqref{eq:pix} and \eqref{eq:qix}, respectively. However, the RN for $\fas$,
instead of $p_i$, uses the following checks \vspace{-5pt}
\begin{equation} 
\label{eq:ppi}
\psix \define \tau p_i(x),\quad \text{ for all } i=1,\ldots,t. \vspace{-5pt}
\end{equation} 
where $\tau$ is chosen so that $\tau \in F \setminus \{0\}$ and $\frac{\tau}{\overline{\alpha} - \alpha^*} \in \Ksb = \spn\big(\{\qoab,\ldots,\qtmoab\}\big)$. For instance, we can set $\tau = (\overline{\alpha} - \alpha^*)\qoab$.
Clearly, $\deg(\psi)=\deg(p_i) < n-k$ and hence, $\psix$ serve as check polynomials of the code. Note that for this scheme to work, the RN for $\fas$ must know $v_1 = \qoab$, so that it can compute $\tau$. This assumption is satisfied as long as every RN uses a known deterministic repair scheme, which should be the case in practice.
 
\begin{lemma} 
\label{lem:ps}
The check polynomials $\psix$ defined as in \eqref{eq:ppi} satisfy the following properties. 
\begin{itemize}
	\item[(P1)] $\psoab = \pstwab = \cdots =  \pstmoab = 0$.
	\item[(P2)] $\big\{\psoas,\ldots,\pstas\big\}$ is a basis of $F$ over $B$. 
	\item[(P3)] $\tr\big(\pstab\fab\big)$ and $\tr\Big(\frac{\tau\fab}{\overline{\alpha} - \alpha^*}\Big)$ are $B$-linear combinations of elements in $\big\{\qoab\fab,\ldots,\qtmoab\fab\big\}$. Moreover, the coefficients in these combinations do not depend on $f$.
\end{itemize}
\end{lemma}  
\begin{proof} 
Property (P1) holds since $p^*_i(\overline{\alpha}) = \tau \piab$ and $\piab = 0$ for all $i \in [t-1]$ by Lemma~\ref{lem:Qp}~(a).  
Property (P2) follows from the fact that $\psias = \tau\pias$, $\tau \neq 0$, 
and that $U' = \{u_1,\ldots,u_t\} = \{\poas,\ldots,\ptas\}$ is a $B$-basis of $F$.  
Property (P3) holds due to the definition of $\tau$ and the linearity of trace. 
\end{proof} 

\nin\textbf{Download Phase.}
In this phase, the RN for $\fas$ uses the first $t-1$ checks $p^*_1,\ldots,p^*_{t-1}$ to construct the $t-1$ repair equations. For $i = 1,\ldots,t-1$, \vspace{-5pt}
\begin{multline}
\label{eq:psi} 
\tr\big(\psias\fas\big) = -\sum_{\alpha \in A \setminus \{\alpha^*\}} \tr\big(\psia\fa\big)\\
=-\sum_{\alpha \in A \setminus \{\alpha^*\}} 
\tr\big(u_i(\alpha-\alpha^*)\big) \times \tr\Big(\frac{\tau\fa}{\alpha-\alpha^*}\Big). \vspace{-5pt}
\end{multline}
By Lemma~\ref{lem:ps}~(P1), the RHS of \eqref{eq:psi} does not involve $\fab$. Thus, the RN for $\fas$ can determine $t-1$ target traces $\tr\big(\psias\fas\big)$, $i \in [t-1]$, of $\fas$
by downloading $n-2$ sub-symbols $\tr\Big(\frac{\tau\fa}{\alpha-\alpha^*}\Big)$ from
the available nodes storing $\fa$, $\alpha \in A \setminus \{\alpha^*, \overline{\alpha}\}$. 
The RN for $\fab$ follows the same procedure as in the first scheme (Section~\ref{subsec:depth_one}).\\ 

\nin\textbf{Collaboration Phase.}
The last repair equation for $\fas$ is \vspace{-8pt} 
\begin{multline}
\label{eq:pst}
\tr\big(\pstas\fas\big) + \tr\big(\pstab\fab\big)\\
=-\sum_{\alpha \in A \setminus \{\alpha^*,\overline{\alpha}\}} \tr\big(\psta\fa\big)\\ 
= -\sum_{\alpha \in A \setminus \{\alpha^*,\overline{\alpha}\}} 
\tr\big(u_t(\alpha-\alpha^*)\big) \times \tr\Big(\frac{\tau\fa}{\alpha-\alpha^*}\Big). 
\end{multline}

Clearly, the RN for $\fas$ can compute the RHS of \eqref{eq:pst} based on what
it downloaded in the Download Phase. To obtain the last target trace $\tr\big(\pstas\fas\big)$, it computes $\tr\big(\pstab\fab\big) = \tr\big(\ut(\overline{\alpha} - \alpha^*)\big)\times\tr\Big(\frac{\tau\fab}{\overline{\alpha} - \alpha^*}\Big)$, where the repair trace $\tr\Big(\frac{\tau\fab}{\overline{\alpha} - \alpha^*}\Big)$ can be downloaded from the RN for $\fab$. This 
is possible because the repair trace is a $B$-linear combination of $\big\{\qoab\fab,\ldots,\qtmoab\fab\big\}$,
due to Lemma~\ref{lem:ps}~(P3), which is already available at the RN for $\fab$. 
Then, by Lemma~\ref{lem:ps}~(P2), the RN for $\fas$ has $t$ independent
traces of $\fas$ to recover this lost symbol. As $\fas$ has been recovered, 
the RN for $\fab$ downloads the repair trace $\tr\Big(\frac{\fas}{\alpha^* -\overline{\alpha}}\Big)$ from the RN for $\fas$ to compute the target trace $\tr\big(\qtab\fab\big)$,
and then can recover $\fab$ completely. 
Note that the RN for $\fas$ has to first receive the repair trace from the RN for $\fab$ before computing and sending out its repair trace for $\fab$ (see Fig.~\ref{fig:distributed_two_depth_two}). 
Theorem~\ref{thm:depth_two} summarizes the discussion. 

\vspace{-5pt}
\begin{figure}[htb]
\centering
\includegraphics[scale=.9]{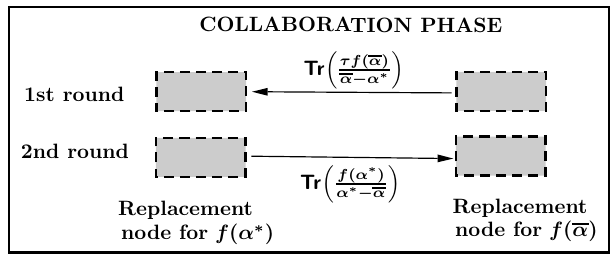}
\caption{Illustration of the Collaboration Phase in the \emph{two-round} distributed repair scheme of two erasures for Reed-Solomon codes. The two repair traces are exchanged \emph{sequentially}.}
\label{fig:distributed_two_depth_two}
\vspace{-10pt}
\end{figure}

\begin{theorem} 
\label{thm:depth_two}
The two-round distributed repair scheme can be used to repair any two erased symbols of a Reed-Solomon codes $\rsk$ at a repair bandwidth of $n-1$ sub-symbols per erased symbol, for any extension degree $t$, given that $n-k \geq |B|^{t-1}$.  
\end{theorem}

\begin{remark}
\label{rm:distributedII}
In our repair scheme, each RN uses a bandwidth of $n-1$ sub-symbols. In a naive scheme, one RN first downloads $kt$ sub-symbols from a set of $k$ available nodes, recovers both erased symbols, and then sends the corresponding symbol to the other RN. The total bandwidth used is $kt+t$, which is worse than the method we described if $\frac{k+1}{n-1} > \frac{2}{t}$, i.e. when $t$ is sufficiently large or when the code has high rate.
\end{remark}

\begin{example}
\label{ex:2}
Let $B = \text{GF}(2)$, $F = \text{GF}(8)$, $n = |F| = 8$, and $k = n(1-1/|B|) = 4$. 
The extension degree is $t = 3$, which is \emph{not} divisible by the characteristic of the fields.
Let $\xi$ be a primitive element of $F$ that satisfies $\xi^3 + \xi + 1 = 0$. Let $A = \{0,1,\xi,\xi^2,\ldots,\xi^{6}\} \equiv F$ be the set of evaluation points and $\C = \rsk$ be the Reed-Solomon code defined as in Definition~\ref{def:RS}.

Let $f(x)$ be a polynomial over $F[x]$ of degree at most $3$. Then $\bc_f
=\big(f(0),f(1),\ldots,f(\xi^{6}) \big)$ is a codeword of $\C$. 
Suppose that the first two codeword symbols $f(0)$ and $f(1)$ are erased.
So $\alpha^* = 0$ and $\overline{\alpha} = 1$. 
We now demonstrate how to construct a two-round distributed repair scheme for
these two erasures. 
Firstly, we have $\qsbz = \tr\big( z(\overline{\alpha} - \alpha^*) \big)
= \tr( z ) = z^4 + z^2 + z$. 
We then choose the bases $U$ and $V$ for the root space 
$\Ksb$ of $\qsbz$ as $U = V = \{\xi,\xi^2\}$.
These sets $U$ and $V$ can be completed to $\ft$-bases $U' = V' = \{\xi,\xi^2,\xi^3\}$ of $F$.
The checks $p_i$ for $f(0)$ and $q_i$ for $f(1)$, $i = 1,2,3$, are constructed as follows. \vspace{-5pt}
\begin{alignat*}{2}
p_1(x) &= \tr\big(\xi(x-0)\big)/(x-0) & &= \xi^4x^3+\xi^2x+\xi,\\
p_2(x) &= \tr\big(\xi^2(x-0)\big)/(x-0) & &= \xi x^3+\xi^4x+\xi^2,\\
p_3(x) &= \tr\big(\xi^3(x-0)\big)/(x-0) & &= \xi^5 x^3 + \xi^6 x + \xi^3\\
q_1(x) &= \tr\big(\xi(x-1)\big)/(x-1) & &= \xi^4x^3 + \xi^4x^2 + \xi x,\\
q_2(x) &= \tr\big(\xi^2(x-1)\big)/(x-1) & &= \xi x^3 + \xi x^2 + \xi^2 x,\\
q_3(x) &= \tr\big(\xi^3(x-1)\big)/(x-1) & &= \xi^5 x^3 + \xi^5 x^2 + \xi x + 1.
\end{alignat*}
Note that in the two-round scheme, we use $p^*_i(x)$ instead of $p_i(x)$. 
Since $\tau = (1-0)q_1(1) = \xi$, we can determine $p^*_i(x)$ as follows. 
\begin{alignat*}{2}
p^*_1(x) &= \xi p_1(x) & &= \xi^5x^3+\xi^3x+\xi^2,\\
p^*_2(x) &= \xi p_2(x) & &= \xi^2 x^3+\xi^5x+\xi^3,\\
p^*_3(x) &= \xi p_3(x) & &= \xi^6 x^3 + x + \xi^4.
\end{alignat*}
The coordinates of the checks are given in Table~\ref{tab:checks}. These coordinates are precisely the coefficients to be used in the repair equations. The reader may refer to the caption of the table
for the explanation of how this particular scheme works. 

\begin{table}[htb]
\normalsize
\centering
\begin{tabular}{|c|c|c|c|c|c|c|c|c|}
\hline
$A$ & $\boldsymbol{\alpha^* = 0}$ & $\boldsymbol{\overline{\alpha} = 1}$ & $\xi$ & $\xi^2$ & $\xi^3$ & $\xi^4$ & $\xi^5$ & $\xi^6$\\
\hhline{|=|=|=|=|=|=|=|=|=|}
$p^*_1$ & $\boldsymbol{\xi^2}$ & $\cdot$ & $\cdot$ & $\xi^6$ & $\cdot$ &  $\xi^4$ & $\xi^3$ & $\xi^2$\\
\hline
$p^*_2$ & $\boldsymbol{\xi^3}$ & $\cdot$ & $1$ & $\cdot$ & $\xi^5$ & $\xi^4$ & $\xi^3$ & $\cdot$\\
\hline
$p^*_3$ & $\boldsymbol{\xi^4}$ & $\boldsymbol{\xi}$ & $\cdot$ & $\xi^6$ & $\xi^5$ & $\xi^4$ & $\cdot$ & $\cdot$\\
\hhline{|=|=|=|=|=|=|=|=|=|}
$q_1$ & $\cdot$ & $\boldsymbol{\xi}$ & $\cdot$ & $\xi$ & $\cdot$ & $\xi^2$ & $\xi^3$ & $\xi^5$\\
\hline
$q_2$ & $\cdot$ & $\boldsymbol{\xi^2}$ & $\xi^4$ & $\cdot$ & $\xi^6$ & $\xi^2$ & $\xi^3$ & $\cdot$\\
\hline
$q_3$ & $\boldsymbol{1}$ & $\boldsymbol{\xi^3}$ & $\xi^4$ & $\cdot$ & $\cdot$ & $\cdot$ & $\xi^3$ & $\xi^5$\\
\hline
\end{tabular}
\vspace{5pt}
\caption{(Example~\ref{ex:2}) The table of dual codewords (checks) used in the \emph{two-round} distributed scheme to repair two erased codeword symbols $f(0)$ and $f(1)$, which correspond to the first two columns. The top row lists the $8$ evaluation points. A dot ``$\cdot$'' means a zero entry. 
The basis $\xi U'=\{\xi^2,\xi^3,\xi^4\}$ of $F = \text{GF}(8)$ over $B = \text{GF}(2)$ corresponds to the evaluations of the three checks $p^*_1,p^*_2,p^*_3$ at $\alpha^* = 0$. Similarly, the basis $V'=\{\xi,\xi^2,\xi^3\}$ corresponds to the evaluations of the three checks $q_1,q_2,q_3$ at $\overline{\alpha} = 1$.
Within the first three rows, except the first two columns, the nonzero symbols in each column are all the same. 
This property reflects the fact that to compute the right-hand side sums in the repair equations generated by $p^*_i$'s, for each $\alpha \neq \{0,1\}$, the RN for $f(0)$ only needs to download one trace (bit) 
$\tr\Big(\frac{\fa}{\alpha} \Big)$ from the node storing $\fa$. Also, the first two rows have all zero entries at the second
column, which means the first two checks $p^*_1$ and $p^*_2$
involve $f(0)$ but not $f(1)$. 
As a consequence, in the Download Phase, the first two independent traces of $f(0)$, $\tr\big(p^*_1(0)f(0)\big)$ and $\tr\big(p^*_2(0)f(0)\big)$, can be determined from $6 = 8 - 2$ downloaded bits.
The third check $p^*_3$ involves \emph{both} $f(0)$ and $f(1)$, as described in the scheme. 
However, because of its definition, $p^*_3(1) = \xi = q_1(1)$. This allows the RN for $f(0)$ to compute the term $\tr\big(p^*_3(1)f(1)\big) = \tr\big(q_1(1)f(1)\big)$ in the repair equation generated by $p^*_3$, by downloading $\tr\big(q_1(1)f(1)\big)$ from the RN for $f(1)$ in the Collaboration Phase.
The RN for $f(0)$ can then extract the third trace
of $f(0)$, namely $\tr\big(p^*_3(0)f(0)\big)$, from the repair equation generated by $p^*_3(x)$. 
It now has three independent traces $\tr\big(p^*_i(0)f(0)\big)$, $i = 1,2,3$, to completely recover the erased symbol $f(0)$. It then sends the repair trace $\tr\Big(\frac{f(0)}{0-1}\Big)$, or equivalently, $\tr\big(q_3(0)f(0)\big)$ to the RN for $f(1)$.  
The RN for $f(1)$ already obtained two independent traces for $f(1)$, namely $\tr\big(q_i(1)f(1)\big)$, $i = 1,2,$ in the Download Phase, by downloading six bits from other available nodes. Once it receives the repair trace from the RN for $f(0)$ in the Collaboration Phase, it can figure out the third trace $\tr\big(q_3(1)f(1)\big)$ and complete the recovery of $f(1)$. The repair bandwidth per erased symbol is seven bits.}
\label{tab:checks}
\vspace{-10pt}
\end{table}
\end{example}

\subsection{Centralized Repair Schemes for Two Erasures}
\label{subsec:centralized}

In a centralized repair scheme, a repair center carries out the recovery of both erased symbols simultaneously at one place. Then it distributes the recovered symbols to the corresponding RNs.
The repair bandwidth is defined to be the total amount of information the repair center downloads from the surviving nodes divided by the number of erased symbols. We can easily obtain a centralized repair scheme from each of the previously developed distributed schemes by 
simply letting the repair center download all the repair traces that both 
the RNs are supposed to download from other available nodes in the Download Phase. As the Collaboration Phase is no longer needed, the amount of
information the repair center downloads from the $n-2$ available nodes is at most $n-2$ sub-symbols in $B$ per erased symbol. 

\begin{theorem} 
\label{thm:centralized}
The centralized repair schemes deduced from the two distributed repair schemes in Section~\ref{subsec:depth_one} and Section~\ref{subsec:depth_two} can recover any two erased symbols $\fas$ and $\fab$ at a repair bandwidth of $n-2$ sub-symbols in $B$ per erased symbol. The first scheme requires that $\cF$ divides $t$, while the second one works for every $t$. 
\end{theorem} 

\begin{remark}
Whenever $2(n-2) < kt$, or equivalently, $\frac{k}{n-2} > \frac{2}{t}$, the centralized schemes developed in this section have a strictly smaller repair bandwidth than the naive scheme.
\end{remark}

\section{Repairing Reed-Solomon Codes With Three Erasures}
\label{sec:three_erasures}

We extend our centralized and distributed schemes developed in Section~\ref{sec:two_erasures} to address the case when \emph{three} codeword
symbols are erased in an RS code. We assume that
$n - k \leq |B|^{t-1}$ and that the field extension degree $t$ is divisible by the characteristic of the field. The unique feature of these repair schemes is the concept of a \emph{repair cycle}, in which the computation of a target trace enables the computation of another target trace, and so forth. The target traces found during the cycle are not known in advance, and are only determined once an \emph{activating} trace is found, based on the approach of Theorem~\ref{thm:three_erasures} and Theorem~\ref{thm:three_erasures_distributed}. 
When two such repair cycles are completed, $t$ independent traces for each of the three erased symbols are obtained. Consequently, 
these symbols can be recovered simultaneously, with a repair cost of $n-3$ downloaded sub-symbols per erased symbol in the centralized scheme, and with a repair cost of $n-1$ downloaded sub-symbols per erased symbol in the distributed scheme. 
At such low bandwidths, our repair schemes can only handle certain patterns of three erasures, and such patterns can be precisely characterized. For other patterns, the repair cycles cannot be activated, and hence, the repair process itself cannot be initiated. Of course, one can always repair those patterns by downloading more traces, leading to larger bandwidths.  

\subsection{A Centralized Three-Erasure Repair Scheme for Reed-Solomon Codes}
\label{subsec:centralized_3erasures}

\begin{theorem} 
\label{thm:three_erasures}
Suppose $f(x) \in F[x]$ is a polynomial of degree at most $k-1$, which
corresponds to a codeword in a Reed-Solomon code $\C = \rsk$. 
Assume that $\fas$, $\fab$, and $\fap$
are the three erased codeword symbols. Given that the field extension 
degree $t$ is divisible by the characteristic of $F$, and that  
\[
\left\{
\frac{\overline{\alpha}-\alpha^*}{\overline{\alpha} - \alpha'},
\frac{\alpha'-\overline{\alpha}}{\alpha' - \alpha^*},
\frac{\alpha^*-\alpha'}{\alpha^* - \overline{\alpha}}
\right\} \cap K \neq \varnothing
\]
holds, where $K = \ker(\tr)$ is the kernel of the trace function, there exists a centralized repair scheme
that recovers the three erased symbols by downloading three sub-symbols from each surviving node. In other words, the repair bandwidth of the scheme is $n-3$ sub-symbols per erased symbol.  
\end{theorem} 

We divide our repair process into two phases. In the \emph{Download Phase}, $s \in \{t-2,t-1\}$ independent target traces for each erased symbol are produced at the cost of downloading a total of $3(n-3)$ sub-symbols from $n-3$ surviving nodes. 
We emphasize that \emph{no} additional data download is needed afterward. In the \emph{Complement Phase}, one or two additional target traces are generated for each erased symbol, depending on whether $s = t-1$ or $s = t-2$. In the more complicated case when $s = t - 2$, the Complement Phase involves two \emph{repair cycles}, each of which produces three target traces. Within each cycle, the determination of one target trace leads to the determination of another target trace.\\ 

\nin\textbf{Download Phase.}
The repair center starts the whole repair process by producing $s \in \{t-2,t-1\}$ independent target traces for \emph{each} erased symbol. Lemma~\ref{lem:intersection} explains the traces construction process. 
Recall that for any three distinct elements $\alpha, \beta, \gamma$ in $F$, one can define the polynomials $\Qabz$, $\Qbgz$, and $\Qgaz$, according to~\eqref{eq:Q}.
Moreover, by \eqref{eq:Kab}, the intersection of the root spaces of \emph{any two} among these three polynomials is
\[
\Kabg \define \left\{z^* \in F \colon \tr(z^*\alpha) = \tr(z^*\beta) = \tr(z^*\gamma)\right\},
\]
which is precisely the intersection of the root spaces of \emph{all} three polynomials. 
We are interested in the case when $\{\alpha,\beta,\gamma\} = \{\alpha^*, \overline{\alpha}, \alpha'\}$.

\begin{lemma} 
\label{lem:intersection}
Let $\alpha$, $\beta$, and $\gamma$ be three distinct elements of $F \cong B^t$. Then
\[
t - 2 \leq \dim_B(\Kabg) \leq t-1.
\] 
\end{lemma} 
\begin{proof}
We have $\Kabg = \Kab \cap \Kbg$.   
From Lemma~\ref{lem:Kab}~(b), $\Kab$ and $\Kbg$ both have dimension $t-1$. Hence,
\[
\dim_B(\Kabg) \leq t-1.
\] 
Moreover, we also have 
\[
K_{\alpha, \beta} = \frac{1}{\beta-\alpha}K,\quad
K_{\beta,\gamma} = \frac{1}{\gamma-\beta}K,
\]
where $K = \ker(\tr)$. 
We may hence write
\[
|\Kabg| = \left|\Big(\dfrac{1}{\beta-\alpha}K\Big) 
\cap \Big(\dfrac{1}{\gamma-\beta}K\Big) \right| 
= \left|\Big(\dfrac{\gamma-\beta}{\beta-\alpha}K\Big) \cap K \right|,
\]
and for $\eta \define \frac{\gamma-\beta}{\beta-\alpha}$, the previous expression reduces to 
\begin{equation}
\label{eq:I}
|\Kabg| = |\eta K \cap K|.
\end{equation} 
Consider the linear mapping $\sigma \colon K \to B$, where $\sigma(\kappa) = \tr(\eta\kappa)$. Then 
\begin{equation} 
\label{eq:sigma}
\ker(\sigma) = \eta K \cap K. 
\end{equation} 
Since $\dim_B(K) = t-1$, from \eqref{eq:I} and \eqref{eq:sigma} we deduce
\[
|\Kabg| = |\ker(\sigma)| \geq |B|^{\dim_B(K)-\dim_B(B)} = |B|^{t - 2}.
\]
Therefore, $\dim_B(\Kabg) \geq t - 2$. 
\end{proof}

According to Lemma~\ref{lem:intersection}, $s \define \dim_B(\Ksbp) \in \{t-2,t-1\}$. 
We now select three $B$-bases $U = \{u_1,\ldots,\us\}$, $V = \{v_1,\ldots,\vs\}$, and $W = \{w_1,\ldots,\ws\}$ of $\Ksbp$. Note that one can set $U$, $V$, and $W$ to the same basis. Here we keep the setting general and assume that they can be any three bases of $\Ksbp$.
Based on these three sets, we may define three types of checks, namely
\begin{itemize}
	\item $p_i(x) \define p_{u_i,\alpha^*}(x) = \tr\big(u_i(x-\alpha^*)\big)/(x-\alpha^*)$, \ $i \in [s]$.
	\item $q_i(x) \define p_{v_i,\overline{\alpha}}(x) = \tr\big(v_i(x-\overline{\alpha})\big)/(x-\overline{\alpha})$, \ $i \in [s]$.
	\item $r_i(x) \define p_{w_i,\alpha'}(x) = \tr\big(w_i(x-\alpha')\big)/(x-\alpha')$, \ $i \in [s]$.
\end{itemize}
From Lemma~\ref{lem:p}~(b), these checks satisfy the following properties.
\begin{itemize}
	\item $\{\poas,\ldots, \psas\} \equiv U$, since $\pias = u_i$,
	\item $\{\qoab,\ldots, \qsab\} \equiv V$, since $\qiab = v_i$,
	\item $\{\roap,\ldots, \rsap\} \equiv W$, since $\riap = w_i$,
\end{itemize}
where $\equiv$ stands for set equivalence. 
Moreover, by Lemma~\ref{lem:Qp}~(a), the following additional properties also hold. 
\begin{itemize}
	\item $\piab = 0$ and $\piap = 0$, for $i \in [s]$. 
	\item $\qias = 0$ and $\qiap = 0$, for $i \in [s]$. 
	\item $\rias = 0$ and $\riab = 0$, for $i \in [s]$. 
\end{itemize}
These $3s$ checks have support patterns as listed in Table~\ref{tab:three_erasures_Download}. We summarize these facts in Lemma~\ref{lem:roundI}.  

\begin{lemma}
\label{lem:roundI}
The following statements hold. 
\begin{itemize}
	\item $p_i(x)$ involves $\fas$, but excludes $\fab$ and $\fap$, $i \in [s]$.
	\item $q_i(x)$ involves $\fab$, but excludes $\fas$ and $\fap$, $i \in [s]$.
	\item $r_i(x)$ involves $\fap$, but excludes $\fas$ and $\fab$, $i \in [s]$.
\end{itemize}
\end{lemma}

\begin{table}[htb]
\centering
\begin{tabular}{|l|c|c|c|l|}\hline
\diagbox[width=10em]{\textbf{Checks}}{\textbf{Erased}\\\textbf{Positions}} & $\alpha^*$ & $\overline{\alpha}$ 
& $\alpha'$ & \textbf{Purpose}\\ 
\hline
$p_1,\ldots,\ps$ & $\times$ & $\cdot$ & $\cdot$ & $\to$ target traces for $\fas$\\ \hline
$q_1,\ldots,\qs$ & $\cdot$ & $\times$ & $\cdot$ & $\to$ target traces for $\fab$\\ \hline
$r_1,\ldots,\rs$ & $\cdot$ & $\cdot$ & $\times$ & $\to$ target traces for $\fap$\\ \hline
\end{tabular}
\caption{The list of checks used in the Download Phase. 
A cross ``$\times$'' denotes a nonzero value, while a dot ``$\cdot$'' denotes a zero value. 
Each of these $3s$ checks involves exactly \emph{one} of the three codeword symbols $\fas$, $\fab$, and $\fap$.}
\label{tab:three_erasures_Download}
\vspace{-5pt}
\end{table}

From these $3s$ checks, the repair center generates $s$ repair equations for each erased symbol as follows. For $i \in [s]$:
\[
\begin{split} 
\tr\big(\pias\fas\big) &= - \sum_{\alpha \in A \setminus \{\alpha^*\}} \tr\big(\pia\fa\big),\\
\tr\big(\qiab\fab\big) &= - \sum_{\alpha \in A \setminus \{\overline{\alpha}\}} \tr\big(\qia\fa\big),\\ 
\tr\big(\riap\fap\big) &= - \sum_{\alpha \in A \setminus \{\alpha'\}} \tr\big(\ria\fa\big). 
\end{split}
\]
Due to the properties of the checks $p_i$, $q_i$, and $r_i$ stated in Lemma~\ref{lem:roundI}, the right-hand side sums of the above repair equations do not involve any of the symbols $\fas$, $\fab$, and $\fap$. Therefore, by Lemma~\ref{lem:trace},
one can determine $s$ independent traces for each erased symbol by downloading three sub-symbols (repair traces) from each surviving node, namely $\tr\Big(\frac{\fa}{\alpha - \alpha^*}\Big)$, $\tr\Big(\frac{\fa}{\alpha - \overline{\alpha}}\Big)$, and
$\tr\Big(\frac{\fa}{\alpha - \alpha'}\Big)$. In summary, after this phase,
the repair center obtains the following target traces for the erased symbols:
\begin{itemize}
	\item $s$ traces $\tr\big(u_i \fas\big)$, $i \in [s]$, for $\fas$,
	\item $s$ traces $\tr\big(v_i \fab\big)$, $i \in [s]$, for $\fab$,
	\item $s$ traces $\tr\big(w_i \fap\big)$, $i \in [s]$, for $\fap$. 
\end{itemize}
\vspace{5pt}

\nin\textbf{Complement Phase.}
By Lemma~\ref{lem:intersection}, $s \in \{t-2,t-1\}$. We consider the following two cases, depending on the value of $s$. 

\underline{Case 1}: $s = t - 2$. The repair center needs two more target traces to complete the recovery of each erased symbol. To that end, \emph{six} additional checks, $\ptmo, \pt$, $\qtmo, \qt$, $\rtmo$, and $\rt$, are constructed, based on the following field elements (see Fig.~\ref{fig:tower}):
\begin{itemize}	
	\item $\utmo \in F$, so that $U \cup \{\utmo\} = \{u_1,\ldots,\utmo\}$ generates the root space $\Ksb$ of $\qsbz$, 
\item $\ut \in F$, so that $U \cup \{\ut\} = \{u_1,\ldots,\utmt,\ut\}$ generates the root space $\Kps$ of $\qpsz$,
	\item $\vtmo \in F$, so that $V \cup \{\vtmo\} = \{v_1,\ldots,\vtmo\}$ generates the root space $\Kbp$ of $\qbpz$,
	\item $\vt \in F$, so that $V \cup \{\vt\} = \{v_1,\ldots,\vtmt,\vt\}$ generates the root space $\Ksb$ of $\qsbz$,
	\item $\wtmo \in F$, so that $W \cup \{\wtmo\} = \{w_1,\ldots,\wtmo\}$ generates the root space $\Kps$ of $\qpsz$,
	\item $\wt \in F$, so that $W \cup \{\wt\} = \{w_1,\ldots,\wtmt,\wt\}$ generates the root space $\Kbp$ of $\qbpz$. 
\end{itemize}

\begin{figure}[htb]
\centering
\includegraphics[scale=1]{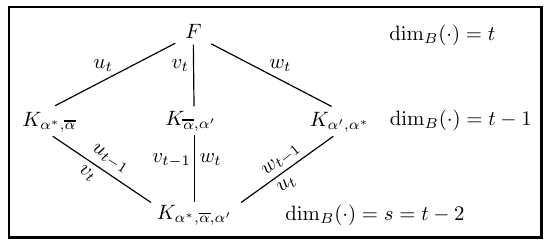}
\caption{The hierarchy of the $B$-subspaces of $F \cong B^t$.}
\label{fig:tower}
\vspace{-5pt}
\end{figure}

Note that by Lemma~\ref{lem:Kab}~(b), the root spaces of all polynomials $Q_{\cdot,\cdot}(z)$ listed above have dimension $t-1$. Since $s = t-2$, 
\begin{multline*}
\dim_B(\Ksb \cap \Kps) = \dim_B(\Ksbp) = t - 2\\
< t-1 = \dim_B(\Ksb) = \dim_B(\Kps),
\end{multline*}
we deduce that $\utmo \notin \Kps$ and $\ut \notin \Ksb$. Therefore, $U' = \{u_1,\ldots,u_t\}$ is a $B$-basis of $F$. Similarly, $V' = \{v_1,\ldots,v_t\}$ and $W' = \{w_1,\ldots,w_t\}$ are also $B$-bases of $F$. 
Thus, the six remaining checks can be defined as
\begin{itemize}
	\item $\ptmo(x) \define p_{\utmo,\alpha^*}(x) = \tr\big(\utmo(x-\alpha^*)\big)/(x-\alpha^*)$,
	\item $\pt(x) \define p_{\ut,\alpha^*}(x) = \tr\big(\ut(x-\alpha^*)\big)/(x-\alpha^*)$,
	\item $\qtmo(x) \define p_{\vtmo,\overline{\alpha}}(x) = \tr\big(\vtmo(x-\overline{\alpha})\big)/(x-\overline{\alpha})$,
	\item $\qt(x) \define p_{\vt,\overline{\alpha}}(x) = \tr\big(\vt(x-\overline{\alpha})\big)/(x-\overline{\alpha})$,
	\item $\rtmo(x) \define p_{\wtmo,\alpha'}(x) = \tr\big(\wtmo(x-\alpha')\big)/(x-\alpha')$,
	\item $\rt(x) \define p_{\wt,\alpha'}(x) = \tr\big(\wt(x-\alpha')\big)/(x-\alpha')$.
\end{itemize}
Lemma~\ref{lem:roundII} concludes that these six checks satisfy the support patterns described in Table~\ref{tab:three_erasures_Complement1}. 

\begin{table}[htb]
\centering
\begin{tabular}{|l|c|c|c|l|}\hline
\diagbox[width=10em]{\textbf{Checks}}{\textbf{Erased}\\\textbf{Positions}} & $\alpha^*$ & $\overline{\alpha}$ 
& $\alpha'$ & \textbf{Purpose}\\ 
\hline
$\ptmo$ & $\times$ & $\cdot$ & $\times$ & $\to$ target trace for $\fas$ \\ \hline
$\pt$ & $\times$ & $\times$ & $\cdot$ & $\to$ target trace for $\fas$\\ \hline
$\qtmo$ & $\times$ & $\times$ & $\cdot$ & $\to$ target trace for $\fab$\\ \hline
$\qt$ & $\cdot$ & $\times$ & $\times$ & $\to$ target trace for $\fab$\\ \hline
$\rtmo$ & $\cdot$ & $\times$ & $\times$ & $\to$ target trace for $\fap$\\ \hline
$\rt$ & $\times$ & $\cdot$ & $\times$ & $\to$ target trace for $\fap$\\ 
\hline
\end{tabular}
\caption{(Case 1) The list of checks used in the Complement Phase when $s = t - 2$. 
A cross ``$\times$'' denotes a nonzero value, while a dot ``$\cdot$'' denotes a zero value.}
\label{tab:three_erasures_Complement1}
\vspace{-5pt}
\end{table}

\begin{lemma}
\label{lem:roundII}
The following statements hold. 
\begin{itemize}
	\item $\ptmo(x)$ involves $\fas$ and $\fap$, but excludes $\fab$.
	\item $\pt(x)$ involves $\fas$ and $\fab$, but excludes $\fap$.
	\item $\qtmo(x)$ involves $\fab$ and $\fas$, but excludes $\fap$.
	\item $\qt(x)$ involves $\fab$ and $\fap$, but excludes $\fas$.
	\item $\rtmo(x)$ involves $\fap$ and $\fab$, but excludes $\fas$.
	\item $\rt(x)$ involves $\fap$ and $\fas$, but excludes $\fab$.
\end{itemize}
\end{lemma}
\begin{proof}
We prove the first statement, which states that $\ptmoas \neq 0$, $\ptmoap \neq 0$, while $\ptmoab = 0$. The other statements can be derived in a similar manner. Indeed, Lemma~\ref{lem:p}~(b) implies that $\ptmoas = \utmo \neq 0$ and Lemma~\ref{lem:Qp}~(a) implies that $\ptmoab = 0$ since $\utmo \in \Ksb$. 
As $\utmo \notin \Kps$, we deduce that $\ptmoap \neq 0$. 
\end{proof}

\begin{figure}[htb]
\centering
\includegraphics[scale=0.83]{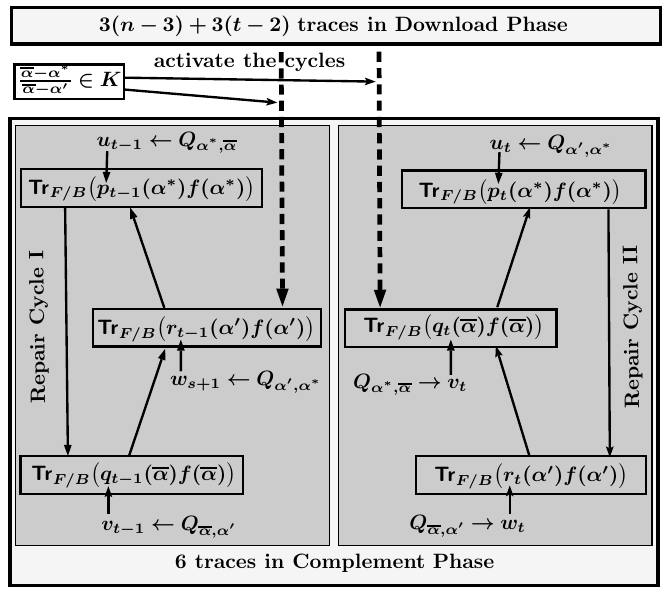}
\caption{Illustration of the centralized repair scheme for three erasures when $s = t-2$.
The long solid arrows within each cycle describe the (cyclic) order in which traces are generated. 
The two dashed arrows indicate that under an additional assumption, namely $\frac{\overline{\alpha}-\alpha^*}{\overline{\alpha} - \alpha'} \in K = \ker(\tr)$, two traces in the two cycles can be generated using the traces obtained during the Download Phase along with their corresponding repair equations.}
\label{fig:repair_cycle}
\vspace{-5pt}
\end{figure}

We outline next the cyclic procedure for generating the six target traces needed for the recovery of the three erased symbols, based on the newly introduced six checks. 
Note that the trace generation within each cycle 
works under the condition that $t$ is divisible by the
characteristic of the fields $B$ and $F$, but no other constraints are needed. However, in order to ``trigger'' the cycles, i.e., in order to generate one of the traces in each cycle, we need to assume further that one of the three ratios listed in Theorem~\ref{thm:three_erasures} must belong to the kernel of the trace function.
The whole repair process is illustrated in Fig.~\ref{fig:repair_cycle}.  

In Lemma~\ref{lem:cycle1} and Lemma~\ref{lem:cycle2} we establish the order in which traces are generated within each cycle.  

\begin{lemma}[Cycle Lemma I] 
\label{lem:cycle1}
Suppose that all $3(n-3)$ repair traces downloaded and $3(t-2)$ target traces constructed in the Download Phase are given. 
Moreover, suppose that $t$ is divisible by $\cF$. Then the following statements hold. 
\begin{itemize}
	\item If the target trace $\tr\big( \rtmoap\fap \big)$ is known, one can determine the target trace $\tr \big( \ptmoas\fas \big)$. 
	\item If the target trace $\tr \big( \ptmoas\fas \big)$ is known, one can determine the target trace $\tr \big( \qtmoab\fab \big)$.
	\item If the target trace $\tr \big( \qtmoab\fab \big)$ is known, one can determine the target trace $\tr\big( \rtmoap\fap \big)$.
\end{itemize} 
\end{lemma} 
\begin{proof} 
Due to symmetry, we only need to prove the first statement. 
Suppose that the target trace $\tr\big( \rtmoap\fap \big)$ of $\fap$ is known.
According to Lemma~\ref{lem:roundII}, $\ptmoab = 0$. Therefore, 
\begin{multline}
\label{eq:ptmo}
\tr \big( \ptmoas\fas \big) + \tr\big( \ptmoap\fap \big)\\
= -\sum_{\alpha \in A \setminus \{\alpha^*, \overline{\alpha}, \alpha'\}}
\tr\big(\ptmoa\fa \big),
\end{multline} 
where the right-hand side sum can be computed from the repair traces obtained in the Download Phase. 
Our purpose is to extract the target trace $\tr\big(\ptmoas\fas \big)$
from this equation, and it suffices to show that we can determine the 
``interfering'' trace $\tr\big( \ptmoap \fap\big)$. 
Recall that $\rtmoap = \wtmo$, and that $\{w_1,\ldots,\wtmo\}$ generates the root space of $\qpsz$. 
Since we already generated $\tr\big( w_i\fap \big)$, for $i \in [t-1]$, by the linearity of the trace function, 
we only need to establish that $\ptmoap = p_{\utmo,\alpha^*}(\alpha')$ is a root of the polynomial $\qpsz$. 
But this claim follows directly from Lemma~\ref{lem:Qp}~(b).  
\end{proof} 

\begin{lemma}[Cycle Lemma II] 
\label{lem:cycle2}
Suppose that all $3(n-3)$ repair traces downloaded and $3(t-2)$ target traces constructed in the Download Phase are given.
Moreover, suppose that $t$ is divisible by $\cF$. Then the following statements hold. 
\begin{itemize}
	\item If the target trace $\tr\big( \qtab\fab \big)$ is known, one can determine the target trace $\tr \big( \ptas\fas \big)$. 
	\item If the target trace $\tr \big( \ptas\fas \big)$ is known, one can determine the target trace $\tr \big( \rtap\fap \big)$.
	\item If the target trace $\tr \big( \rtap\fap \big)$ is known, one can determine the target trace $\tr\big( \qtab\fab \big)$.
\end{itemize} 
\end{lemma} 
\begin{proof}
The proof of this lemma proceeds along the same lines as the proof of Lemma~\ref{lem:cycle1}.  
\end{proof} 

To establish the validity of the procedure used in the Complement Phase, it remains to show that one can simultaneously ``activate'' 
the two cycles whenever one of the ratios of pairwise differences among $\alpha^*$, $\overline{\alpha}$, and $\alpha'$ belongs to the kernel $K$ of the trace function. 

\begin{lemma}[Activation Lemma]
\label{lem:activation}
Suppose that all $3(n-3)$ repair traces downloaded and $3(t-2)$ target traces constructed in the Download Phase are given. 
Moreover, suppose that $t$ is divisible by $\cF$. Then the following statements hold. 
\begin{itemize}
	\item If $\frac{\overline{\alpha}-\alpha^*}{\overline{\alpha} - \alpha'} \in K$, then the target traces $\tr\big( \rtmoap\fap \big)$ and $\tr\big( \qtab\fab \big)$ can be computed from the traces determined in the Download Phase.
	\item If $\frac{\alpha'-\overline{\alpha}}{\alpha' - \alpha^*} \in K$, then the target traces $\tr\big( \ptmoas\fas \big)$ and $\tr\big(\rtap\fap \big)$ can be computed from the traces determined in the Download Phase.
	\item If $\frac{\alpha^*-\alpha'}{\alpha^*-\overline{\alpha}} \in K$, then the target traces $\tr\big( \qtmoab\fab \big)$ and $\tr\big(\ptas\fas \big)$ can be computed from the traces determined in the Download Phase.
\end{itemize}
\end{lemma} 
\begin{proof} 
Due to symmetry, it suffices to prove the first statement only. 
Suppose that $\frac{\overline{\alpha}-\alpha^*}{\overline{\alpha} - \alpha'} \in K$. We first show that the target trace
$\tr\big( \rtmoap\fap \big)$ can be determined based on the $t-2$ target traces
of $\fab$ and the $n-3$ repair traces for $\fap$ obtained in the Download Phase. Since $\rtmoas = 0$ due to Lemma~\ref{lem:roundII}, we can write the repair equation generated by $\rtmo(x)$ as follows:
\begin{multline}
\tr\big(\rtmoap\fap \big) + \tr\big( \rtmoab\fab \big)\\
= -\sum_{\alpha \in A \setminus \{\alpha^*, \overline{\alpha}, \alpha'\}}
\tr\big(\rtmoa\fa \big).
\end{multline} 
The right-hand side sum can be determined based on the repair traces for $\fap$
downloaded in the Download Phase. In order to extract the target trace $\tr\big( \rtmoap\fap \big)$, we need to cancel out the ``interfering'' trace
$\tr\big( \rtmoab\fab \big)$ from this repair equation. To do this, we show that
the ``interfering'' trace can be written as a $B$-linear combination of the
$t-2$ target traces of $\fab$, namely
\[
\tr\big( v_1\fab\big),\ldots,\tr\big(v_{t-2}\fab \big),
\]
which are obtained in the Download Phase. Note that 
$V = \{v_1,\ldots,v_{t-2}\}$ is a basis of $\Ksbp$, which is precisely the intersection of the root spaces of $\qbpz$ and $\qsbz$. Therefore, it suffices to show that $\rtmoab$ is a root of both $\qbpz$ and $\qsbz$. By Lemma~\ref{lem:Qp}~(b), $\rtmoab = p_{\wtmo,\alpha'}(\overline{\alpha})$ is a root of $\qbpz$. Hence, it remains to prove that $\rtmoab$ is a root of $\qsbz$. 

We now invoke the assumption that
$\frac{\overline{\alpha}-\alpha^*}{\overline{\alpha} - \alpha'} \in K$. This assumption implies that 
$\overline{\alpha}-\alpha^* = \kappa(\overline{\alpha} - \alpha')$, for some $\kappa \in K$. Let $\Delta = \overline{\alpha} - \alpha'$ and $b = \tr\big(\wtmo(\overline{\alpha} - \alpha') \big) \in B$. Then
\[
\rtmoab = \tr\big(\wtmo(\overline{\alpha} - \alpha')\big)/(\overline{\alpha} - \alpha') = b / \Delta,
\]
and since $\overline{\alpha}-\alpha^* = \kappa(\overline{\alpha} - \alpha') = \kappa \Delta$, we obtain
\[
\qsbz = \tr\big(z (\overline{\alpha}-\alpha^*) \big)
= \tr(z\kappa \Delta ).
\]
Therefore,
\[
\qsb\big( \rtmoab \big) = \tr\big((b/\Delta)\kappa \Delta \big)
= b\tr(\kappa)
= 0,
\]
because $b \in B$ and $\kappa \in K = \ker(\tr)$. Thus, $\rtmoab$ is a 
root of $\qsbz$ as desired. 

Using a similar approach, we show that the target trace $\tr\big(\qtab\fab \big)$ can be determined as well. We use the repair equation corresponding to $\qt(x)$,
keeping in mind that due to Lemma~\ref{lem:roundII}, $\qtas = 0$:  
\[
\tr\big(\qtab\fab \big) + \tr\big( \qtap\fap \big)
= \sum_{\alpha \in A \setminus \{\alpha^*, \overline{\alpha}, \alpha'\}}
\tr\big(\qta\fa \big).
\] 
Again, the idea is to show that the ``interfering'' term
$\tr\big(\qtap\fap \big)$ can be written as a $B$-linear combination of the $t-2$ target traces of $\fap$, namely 
\[
\tr\big( w_1\fap\big),\ldots,\tr\big(w_{t-2}\fap \big),
\]
which were already generated in the Download Phase. Note that $W = \{w_1,\ldots,w_{t-2}\}$, and by its definition, represents a basis of $\Ksbp$, which is the intersection
of the root spaces of $\qbpz$ and $\qpsz$. Therefore, it suffices to
show that $\qtap$ is a root of both polynomials. The fact that
$\qtap = p_{\vt,\overline{\alpha}}(\alpha')$ is a root of $\qbpz$ follows from Lemma~\ref{lem:Qp}~(b). 
To show that $\qtap$ is also a root of $\qpsz$, note that
\[
\dfrac{\alpha^* - \alpha'}{\alpha' - \overline{\alpha}}
= \dfrac{\overline{\alpha}-\alpha^*}{\overline{\alpha} - \alpha'} - 1 = 
(\kappa - 1) = \kappa' \in K,
\]   
since $1 \in K$ (from the assumption that $t$ is divisible by $\cF$). 
Therefore, if we set $b' = \tr\big(\vt(\alpha' - \overline{\alpha}) \big) \in B$, then 
\[
\qps\big(\qtap\big) = \tr\Big(\frac{b'}{\alpha'-\overline{\alpha}}(\alpha^* - \alpha') \Big)\\ 
= b'\tr(\kappa')
= 0.
\]
The proof follows. 
\end{proof} 

Lemma~\ref{lem:cycle1}, Lemma~\ref{lem:cycle2}, and Lemma~\ref{lem:activation} complete the analysis of the case $s=t-2$ of the Complement Phase. 
At the end of the Complement Phase, the repair center has obtained $t$ independent traces for each erased symbol and therefore, is capable of recovering all three 
lost symbols. It remains to consider the case when $s = t - 1$. 

\underline{Case 2}: $s = t-1$. In this case, the roots spaces of the three polynomial
$\qsbz$, $\qbpz$, and $\qpsz$ coincide. This happens, for example, when $\frac{\overline{\alpha} - \alpha^*}{\alpha^*-\alpha'} \in B$. In such a case we get 
\[
\ker(\qsb) = \frac{1}{\overline{\alpha} - \alpha^*}K \equiv \frac{1}{\alpha^*-\alpha'}K = \ker(\qps).
\]
In this case, for each erased symbol, as the repair center already obtained $s = t - 1$ target traces in the Download Phase, one more target trace is needed for the recovery of the symbol.

We construct the last three checks $\pt(x)$, $\qt(x)$, and $\rt(x)$ via the field elements $\ut$, $\vt$, and $\wt$, chosen as follows:
\begin{itemize}
	\item Choose $\ut$ so that $\rank\big(\{u_1,\ldots,\ut\}\big) = t$.
	\item Choose $\vt$ so that $\rank\big(\{v_1,\ldots,\vt\}\big) = t$.
	\item Choose $\wt$ so that $\rank\big(\{w_1,\ldots,\wt\}\big) = t$.
\end{itemize}
We set 
\[
\begin{split}
\pt(x) &\define p_{\ut,\alpha^*}(x) = \tr\big( \ut(x-\alpha^*) \big)/(x - \alpha^*),\\
\qt(x) &\define p_{\vt,\overline{\alpha}}(x) = \tr\big( \vt(x-\overline{\alpha}) \big)/(x - \overline{\alpha}),\\
\rt(x) &\define p_{\wt,\alpha'}(x) = \tr\big( \wt(x-\alpha') \big)/(x - \alpha').
\end{split}
\]
Table~\ref{tab:three_erasures_Complement2} summarizes the different types of checks and their uses in the repair process when $s = t - 1$. 
\begin{table}[htb]
\centering
\begin{tabular}{|l|c|c|c|l|}\hline
\diagbox[width=10em]{\textbf{Checks}}{\textbf{Erased}\\\textbf{Positions}} & $\alpha^*$ & $\overline{\alpha}$ 
& $\alpha'$ & \textbf{Purpose}\\ 
\hline
$\pt$ & $\times$ & $\times$ & $\times$ & $\to$ target trace for $\fas$\\ \hline
$\qt$ & $\times$ & $\times$ & $\times$ & $\to$ target trace for $\fab$\\ \hline
$\rt$ & $\times$ & $\times$ & $\times$ & $\to$ target trace for $\fap$\\ 
\hline
\end{tabular}
\caption{(Case 2) The list of checks used in the Complement Phase when $s = t - 1$. 
A cross ``$\times$'' denotes a nonzero value, while a dot ``$\cdot$'' denotes a zero value.}
\label{tab:three_erasures_Complement2}
\end{table}

Each of these three checks involve all three erased symbols. Hence, in each repair equation, there are two ``interfering'' terms. However, based on 
Lemma~\ref{lem:last_traces}, we can determine all of these terms from the repair and target traces obtained in the Download Phase. Consequently, the traces $\tr\big( \ut\fas \big)$, $\tr\big( \vt\fab \big)$,
and $\tr\big( \wt\fap \big)$, can also be computed. 
Note that no additional symbol download is required. 

\begin{lemma} 
\label{lem:last_traces}
Given $3t-3$ target traces of $\fas$, $\fab$, and $\fap$, and $3(n-3)$ repair traces obtained in the Download Phase, the target traces 
$\tr\big( \ut\fas \big)$, $\tr\big( \vt\fab \big)$,
and $\tr\big( \wt\fap \big)$ can be determined as long as $t$ is divisible by $\cF$.
\end{lemma} 
\begin{proof} 
By symmetry, it suffices to show that the target trace $\tr\big( \ut\fas \big)$ can be determined using the known traces. The repair equation generated by $\pt(x)$
is given as follows. 
\begin{multline}
\label{eq:pt_case2}
\tr\big( \ptas\fas \big) + \tr\big( \ptab\fab \big) + \tr\big( \ptap\fap \big)\\
= -\sum_{\alpha \in A \setminus \{\alpha^*, \overline{\alpha}, \alpha'\}}
\tr\big( \pta \fa \big).
\end{multline}
Note that the right-hand side sum of this repair equation can be determined using the $n-3$ repair traces $\tr\Big(\frac{\fa}{\alpha - \alpha^*}\Big)$ retrieved in the Download Trace. 

We now prove that the two ``interfering'' traces $\tr\big( \ptab\fab \big)$
and $\tr\big( \ptap\fap \big)$ can be deduced from the known target traces for $\fab$ and $\fap$, respectively. 
Indeed, by Lemma~\ref{lem:Qp}~(b), $\ptab = p_{\ut,\alpha^*}(\overline{\alpha})$ belongs to the root space of $\qsbz$, which
is generated by $\{v_1,\ldots,\vtmo\}$. Therefore, the ``interfering'' trace $\tr\big( \ptab\fab \big)$ is a $B$-linear combination of the known target traces $\tr\big( v_i\fab \big)$, $i \in [t-1]$. Similarly, as $\ptap = p_{\ut,\alpha^*}(\alpha')$ belongs to the root space of $\qpsz$, which is generated by $\{w_1,\ldots,w_{t-1}\}$, the ``interfering'' trace $\tr\big( \ptap\fap \big)$ is a $B$-linear combination of the known target traces
$\tr\big( w_i\fap \big)$, $i \in [t-1]$. This completes the proof. 
\end{proof} 

\subsection{A Distributed Three-Erasure Repair Scheme for Reed-Solomon Codes}
\label{subsec:distributed_3erasures}

We can easily modify the centralized repair scheme proposed in Section~\ref{subsec:centralized_3erasures} to obtain a distributed scheme. The following theorem is the distributed version of Theorem~\ref{thm:three_erasures}. 

\begin{theorem} 
\label{thm:three_erasures_distributed}
Suppose $f(x) \in F[x]$ is a polynomial of degree at most $k-1$, which
corresponds to a codeword in an Reed-Solomon code $\C = \rsk$. 
Assume that $\fas$, $\fab$, and $\fap$
are the three erased codeword symbols. Given that the field extension 
degree $t$ is divisible by the characteristic of $F$, and that 
\[
\left\{
\frac{\overline{\alpha}-\alpha^*}{\overline{\alpha} - \alpha'},
\frac{\alpha'-\overline{\alpha}}{\alpha' - \alpha^*},
\frac{\alpha^*-\alpha'}{\alpha^* - \overline{\alpha}}
\right\} \cap K \neq \varnothing
\]
holds, where $K = \ker(\tr)$ is the kernel of the trace function, there exists a distributed repair scheme where each lost symbols is recovered with a repair bandwidth of $n-1$ sub-symbols.  
\end{theorem} 

\begin{remark}
\label{rm:distributedIII}
In a naive scheme, one RN first downloads $kt$ sub-symbols from a set of $k$ available nodes, recovers all three erased symbols, and then sends the corresponding symbols to the other two RNs. The total bandwidth used is $kt+2t$, which is worse than the scheme we proposed if $\frac{k+2}{n-1} > \frac{3}{t}$, i.e. when $t$ is sufficiently large or when the code has high rate.
\end{remark}

We use the same notations as in Section~\ref{subsec:centralized_3erasures}. 
The repair process also consists of two phases. \\

\nin\textbf{Download Phase.}
\begin{itemize}
	\item The RN for $\fas$ downloads $n-3$ repair traces $\tr\Big(\frac{\fa}{\alpha - \alpha^*}\Big)$ from $n-3$ available nodes, where $\alpha \in A \setminus \{\alpha^*, \overline{\alpha}, \alpha'\}$, and recovers $s$ independent target traces of $\fas$, i.e. $\tr\big(u_i \fas\big)$, $i \in [s]$.
	\item The RN for $\fab$ downloads $n-3$ repair traces $\tr\Big(\frac{\fa}{\alpha - \overline{\alpha}}\Big)$ from $n-3$ available nodes, where $\alpha \in A \setminus \{\alpha^*, \overline{\alpha}, \alpha'\}$, and recovers $s$ independent target traces of $\fab$, i.e. $\tr\big(v_i \fab\big)$, $i \in [s]$.
	\item The RN for $\fap$ downloads $n-3$ repair traces $\tr\Big(\frac{\fa}{\alpha - \alpha'}\Big)$ from $n-3$ available nodes, where $\alpha \in A \setminus \{\alpha^*, \overline{\alpha}, \alpha'\}$, and recovers $s$ independent target traces of $\fap$, i.e. $\tr\big(w_i \fap\big)$, $i \in [s]$.	
\end{itemize}
 
\nin\textbf{Collaboration Phase.}

We describe the tasks the RN for $\fas$ performs in this phase.
Similar descriptions apply to the RNs for $\fab$ and $\fap$. 
We consider two cases, depending on whether $s = t-2$ or $s = t-1$.  

\underline{Case~1}: $s = t - 2$. The RN for $\fas$ computes
the last two target traces $\tr\big(\utmo \fas\big)$ and $\tr\big(\ut \fas\big)$ as follows.
The first trace can be obtained by using the repair equation \eqref{eq:ptmo} generated by $\ptmo$, in which the only unknown term $\tr\big( \ptmoap \fap\big)$ can be determined by downloading an additional repair trace $\tr\Big( \frac{\fap}{\alpha'-\alpha^*} \Big)$ from the RN for $\fap$. The availability of this repair trace at the RN for $\fap$ is explained by either Cycle Lemma~I (Lemma~\ref{lem:cycle1}) or Activation Lemma (Lemma~\ref{lem:activation}).
A similar assertion holds for the second trace.  
The RN for $\fas$ now has $t$ independent traces to recover this lost symbol.
It has downloaded in total $n-1 = (n - 3) + 2$ sub-symbols (repair traces) from all other $n-1$ nodes.  
We illustrate the flow of the recovery data among the RNs in Fig.~\ref{fig:distributed_three_case1}, when $\frac{\overline{\alpha}-\alpha^*}{\overline{\alpha} - \alpha'} \in K$.
Essentially, if a RN can determine a new target trace, it can also compute the repair trace needed for the next RN along the cycle, and so forth.

\begin{figure}[htb]
\centering
\includegraphics[scale=1]{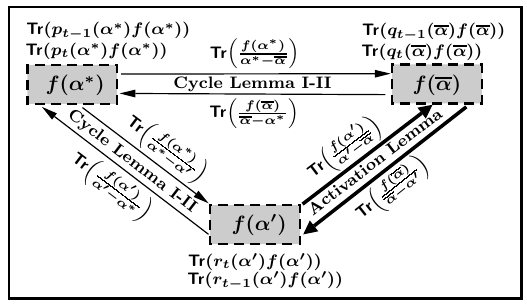}
\caption{Illustration of the data flow in the Collaboration Phase when $s=t-2$ and $\frac{\overline{\alpha}-\alpha^*}{\overline{\alpha} - \alpha'} \in K$.
The two repair cycles are first activated at the replacement nodes for $\fab$
and $\fap$, which can compute new target traces from their previously obtained traces, thanks to the Activation Lemma.
These two nodes then are able to compute the repair traces to send to the next
RN along the two cycles. This process is repeated until every node is able to 
recover the two new target traces. The outer arrows represent the Repair Cycle I, while the inner arrows represent the Repair Cycle II.
The two cycles operate independently in parallel. 
}
\label{fig:distributed_three_case1}
\end{figure}  

\underline{Case 2}: $s = t-1$. The RN for $\fas$ can compute
the last target trace $\tr\big(\ut \fas\big)$ based on the repair equation
\eqref{eq:pt_case2} generated by $\pt$. The only unknown terms are $\tr\big( \ptab\fab \big)$ and $\tr\big( \ptap\fap \big)$.  
According to the proof of Lemma~\ref{lem:last_traces}, these two terms can
be determined by downloading two additional sub-symbols (repair traces)
$\tr\Big( \frac{\fab}{\overline{\alpha} -\alpha^*} \Big)$ and 
$\tr\Big( \frac{\fap}{\alpha' -\alpha^*} \Big)$ from the RNs for $\fab$
and $\fap$, respectively. Thus, the total number of downloaded sub-symbols is $n-1 = (n-3) + 2$.   

\subsection{Correctable Patterns of Three Erasures}

We next evaluate the number of three-erasure patterns that are correctable at low costs
by our centralized and distributed repair schemes developed in the previous sections. 

Under the assumption that the field extension degree
$t$ is divisible by the characteristic of the fields $B$ and $F$, according to Theorem~\ref{thm:three_erasures}, the erased codeword symbols corresponding to three evaluation points
$\alpha$, $\beta$, and $\gamma$, that satisfy   
\begin{equation} 
\label{eq:abg}
\left\{
\frac{\beta-\alpha}{\beta-\gamma},
\frac{\gamma-\beta}{\gamma-\alpha},
\frac{\alpha-\gamma}{\alpha-\beta}
\right\} \cap K \neq \varnothing,
\end{equation} 
can always be recovered by downloading $3(n-3)$ sub-symbols from the $n-3$ unerased codeword symbols. 
Fix two arbitrary evaluation points, say $\alpha$ and $\beta$. We would like to know how many choices there are for the third point $\gamma$ so that the
triple of points satisfies \eqref{eq:abg}. 
For two fixed distinct elements $\alpha$ and $\beta$, and for each $\kappa \in K \setminus \{0,1\}$, we can define 
\[
\gamma = \beta - \frac{\beta-\alpha}{\kappa}.
\]
It is straightforward to verify that $\gamma \neq \alpha$, $\gamma \neq \beta$, and
that $\frac{\beta-\alpha}{\beta-\gamma} \in K$. Therefore, there are $|K|-2 = \frac{n}{|B|}-2$ choices for $\gamma$ that ensure that 
$\frac{\beta-\alpha}{\beta-\gamma} \in K$. 
The same assertion holds if we were to require $\frac{\gamma-\beta}{\gamma-\alpha} \in K$ or 
$\frac{\alpha-\gamma}{\alpha-\beta} \in K$. Hence, the number of correctable triples is at least
a fraction of $\frac{1}{|B|}$ of the number of all triples. 


A computer-aided count of the number of correctable triples, two of which are fixed, is given in Table~\ref{tab:counting} for  some small values of $|B|$ and $t$. 

\begin{table}[htb]
\centering
\begin{tabular}{|l|c|c|}
\hline
& $\#$Correctable triples & $\#$All triples $=n - 2$\\ \hline
$|B|=2,t=4$ & $14$ & $14$\\ \hline	
$|B|=2,t=6$ & $60$ & $62$\\	\hline
$|B|=2,t=8$ & $206$ & $254$\\	\hline
$|B|=2,t=10$ & $900$ & $1022$\\	\hline
$|B|=4,t=4$ & $158$ & $254$\\	\hline
$|B|=4,t=6$ & $2,330$ & $4,094$\\ \hline
$|B|=4,t=8$ & $37,886$ & $65,534$\\ \hline
$|B|=8,t=4$ & $1,406$ & $4,094$\\ \hline
$|B|=8,t=6$ & $86,694$ & $262,142$\\ \hline
$|B|=3,t=3$ & $19$ & $25$\\ \hline
$|B|=3,t=6$ & $529$ & $727$\\ \hline
$|B|=3,t=9$ & $14,083$ & $19,681$\\ \hline
$|B|=9,t=3$ & $223$ & $727$\\ \hline
$|B|=9,t=6$ & $158,263$ & $531,439$\\ 
\hline
\end{tabular}
\caption{Fixing two erased locations, the second column counts the number of the 
third erased location that result in correctable patterns of three erasures. 
The third column specifies the number of all possible choices for the third location
once the first two are fixed, which is $n-2$.}  
\label{tab:counting}
\vspace{-10pt}
\end{table}  

\section{Conclusions}
\label{sec:conclusions}

\subsection{Open Problems}
\label{subsec:open_problems}

We proposed centralized and distributed repair schemes for the recovery of
two and three erasures/failures in Reed-Solomon codes. The distributed schemes offer the same repair bandwidths per erasure as in the case of a single erasure. 
Several open questions remain, including
\begin{itemize}
	\item[(P1)] Finding distributed repair schemes that can recover \emph{all} three-erasure patterns with the repair bandwidths of $n-1$ sub-symbols (per erased symbol) and, in addition, when $t$ is not divisible by $\cF$.
	\item[(P2)] Establishing a lower bound on the repair bandwidth and developing repair schemes for an arbitrary number of erasures that meet the bound. We note that in a follow-up work, Bartan and Wootters~\cite{BartanWootters_MultipleErasures_2017} extended our \emph{centralized} repair scheme to tackle up to $n-k$ erasures. They were also able to improve the repair bandwidth of the scheme based on the observation that the downloaded data in our current scheme is often redundant. Such an extension in the \emph{distributed} setting remains unknown.
	\item[(P3)] Developing efficient repair schemes for codes with an arbitrary number of parities, with particular emphasis on those used in practice, which usually have short lengths, small redundancy, and $t = 8,16,32$, which is closer to $n$ than $\log(n)$ as considered in the current setting. We remark that Duursma and Dau~\cite{DuursmaDau2017} recently obtained a reduction of 32.5\% in repair bandwidth of RS(14,10) code that is used in Facebook f4 Storage System, which improved upon previous reductions of 18.75\% in~\cite{Shanmugam2014} and 20\% in~\cite{GuruswamiWootters2016}, for single erasures. It is of great practical interest to investigate the possibility of repairing such codes (see Table~\ref{tab:popular_systems}) more efficiently in the presence of multiple erasures.
	\item[(P4)] In the large subpacketization regime, Ye and Barg~\cite{YeBarg_RS_Multi_2017} extended the work of Tamo {\et}~\cite{TamoYeBarg2017} to construct RS codes that adopt \emph{centralized} repair schemes with bandwidths achieving the cut-set bound (\cite{Dimakis_etal2007, Dimakis_etal2010}) for any number of erasures. Investigating \emph{distributed} repair schemes for RS codes in 
that regime is another open problem for future research.
\end{itemize}

\subsection{Related Works and Performance Comparison}

Apart from those recent works on repairing RS codes mentioned in Section~\ref{subsec:open_problems}, there is also a number of related results in the literature that extend minimum storage regenerating (MSR) codes~\cite{Dimakis_etal2007, Dimakis_etal2010}, originally designed to recover single erasures only, to cope with multiple erasures. 
In such settings, one often fixes $n$, $k$, and a finite field $B'$, and constructs an MDS code over a larger field $F' \cong B'^{s}$. The field extension degree $s$ is referred to as the \emph{subpacketization} level. In most of these works, $s$ is often larger than or equal to $n-k$. By contrast, in our work, we fix $n$, $k$, and $F$, and choose a subfield $B$ of $F$ so that $|F| = |B|^t$ and $t \leq \log_{|B|}(n-k) \ll n-k$.
By slightly adapting the notation of~\cite{TamoWangBruck_ISIT2011}, we may define the \emph{rebuilding ratio} as the fraction of data (originally) stored at \emph{each} helper node that is sent to a RN, or to the repair center, in order to reconstruct the content of \emph{each} failed node. Most of the codes known in the literature have a rebuilding ratio of $1/(n-k)$, while Reed-Solomon codes with our repair schemes have a rebuilding ratio of $1/t \geq 1/\big(\log_{|B|}(n-k)+1\big)$. MDS codes with larger subpacketization levels allow for smaller rebuilding ratios at the cost
of increasing encoding/decoding as well as code description complexities. 

We briefly review these results and compare them with the results described in the previous sections.  
 
The first line of work is concerned with the problem of repairing $e$ erasures $(e \geq 1)$ in MDS codes in a \emph{distributed} manner. The distributed repair process is often referred to as \emph{cooperative} or \emph{collaborative} repair. 
A cooperative repair scheme generally consists of two phases. In the first phase, the RNs contact and download recovery data from the available nodes. In the second phase, they exchange information in order to help each other to complete their repair processes. Our distributed repair schemes may be viewed as an adaptation of this general strategy to the trace-repair framework~\cite{GuruswamiWootters2016}.  

Hu~{\et}~\cite{Hu2010} presented a lower bound on the total repair bandwidth, which equals $\frac{(n-1)es}{n-k}$ sub-symbols in a subfield $B'$, where $s$ is the subpacketization level.
Wang~{\et}~\cite{Wang2010} generalized this bound to $\sum_{i=1}^e\frac{d_is}{d_i-k+1}$, by allowing each RN for $f(\alpha_i)$ to connect to $d_i$ nodes (including both available and RNs), where $d_i \leq n-1$. MDS codes with repair bandwidths attaining this bound are referred to as \emph{minimum-storage cooperative regenerating (MSCR) codes}.  
The authors of~\cite{Hu2010,Wang2010} also provided probabilistic constructions of such MSCR codes with $s = n-k$. These codes ensure \emph{functional} instead of \emph{exact} repair, i.e. there is no guarantee that the repaired content of a failed node is the same as the original. 
When $d_i=n-1$ for every $i \in [n]$, these MSCR codes have a rebuilding ratio of $\frac{1}{n-k}$. 
Our distributed repair schemes for Reed-Solomon codes with $e = 2,3$ (see Remark~\ref{rm:distributedII}) achieve a building ratio $\frac{1}{t} \geq \frac{1}{\log_{|B|}(n-k)+1}$.  Thus, their random codes have a lower rebuilding ratio than the Reed-Solomon codes under our repair schemes, while employing a higher level of subpacketization.  
Le Scouarnec~\cite{Scouarnec2012} proposed another construction of MSCR codes with \emph{exact} repair for $e=2$ erasures when $k = 2$. This result was later extended to cover all $k \geq 3$ and $n = 2k$ in the work of Chen and Shum~\cite{ChenShum2013}. 
Exact-repair MSCR codes also exist for $d = k$ and $e \leq n - d$, with $s = n - k$, as shown by Shum and Hu~\cite{ShumHu2013}.
Li and Li~\cite{LiLi2014} proved the interesting fact that any MSR code~(see~\cite{Dimakis_etal2007,Dimakis_etal2010}) 
is also an MSCR code for $e = 2$ erasures. Therefore, an MSCR code for $e = 2$ exists whenever an MSR code with the same parameters exists.  
The most recent work along this line is the paper by Shum and Chen~\cite{ShumChen2016}, which introduced a repair scheme that can recover any number of \emph{systematic} node failures in the MISER code~\cite{Shah2012}, for which $n = 2k$. This scheme is also bandwidth-optimal, and uses a level of subpacketization equal to $n - k$.      

The problem of \emph{centralized} repair of multiple erasures for MDS codes has also been studied in the literature. 
Cadambe~{\et}~\cite{Cadambe2013} proved that in order to repair $e$ erasures, the total repair bandwidth has to be 
at least $\frac{eds}{d+e-k}$ sub-symbolss from a subfield $B'$, where $d \leq n-e$ is the number of available nodes that the repair center contacts. 
MDS codes with repair bandwidths attaining this bound are referred to as \emph{minimum-storage multi-node regenerating (MSMR) codes}. 
They also showed the existence of \emph{asymptotic} MSMR codes with  
total repair bandwidths approaching the lower bound when the subpacketization level $s \to \infty$, given that the subfield $B'$ is fixed. 
When $d = n - e$, Tamo~{\et}~\cite{TamoWangBruck2013} demonstrated that their Zigzag code is an exact-repair MSMR code if restricted to $e$ \emph{systematic} node erasures, with subpacketization level $s = (n-k)^{k-1}$. 
Recently, this result was extended by Wang~{\et}~\cite{WangTamoBruck2016} to cover \emph{all} $e$-erasure patterns. Note that when $d = n - e$, the rebuilding ratio of an MSMR code is $\frac{1}{n-k}$, which is smaller than the rebuilding ratio of $\frac{1}{t}
\geq \frac{1}{\log_{|B|}(n-k)+1}$ of our centralized repair schemes for Reed-Solomon codes with $e = 2,3$. However, their subpacketization level $s=(n-k)^{n-1}$ is also much larger than ours, which equals $t \leq \log_{|B|}(n-k)+1$. 
The centralized repair problem was also investigated in the context of wireless distributed storage
systems, where storage nodes are fully connected by a common broadcast channel~\cite{HuSungChan2015}. 

Ye and Barg~\cite{YeBarg_ISIT2016,YeBarg2016} introduced a new notion of 
MDS codes with the \emph{universally error-resilient} $(e,d)$-optimal repair property.
These codes are MSMR codes for all $e \leq n-k$ and all $d \leq n - e$ simultaneously.
The authors of~\cite{YeBarg_ISIT2016,YeBarg2016} also provided a construction of such codes 
whenever the subfield $B'$ has size $|B'| \geq sn$, where $s = \lcm(1,2,\ldots,n-k)$, with subpacketization level $t = s^n$. 

\section*{Acknowledgment}
This work has been supported in part by the NSF grant CCF 1526875, the Center for Science of Information under the grant NSF 0939370, and the NSF grant CCF 1619189. 
H. M. Kiah was also supported in part by the Singapore Ministry of Education under Research Grants MOE2016-T1-001-156 and MOE2015-T2-2-086.

\bibliographystyle{IEEEtran}
\bibliography{RepairingRSCodes_MultipleErasures}

\begin{IEEEbiographynophoto}
{Hoang Dau} received the B.S. degree in applied
mathematics and informatics from Vietnam National
University, Hanoi, Vietnam, in 2006, and the M.S.
and Ph.D. degrees in mathematical sciences from
Nanyang Technological University, Singapore, in
2009 and 2012, respectively.
He is currently an ARC DECRA Fellow and Research Fellow 
with the Department of Electrical and Computer Systems Engineering, Monash University. His research interests include coding theory, network coding, and combinatorics.
\end{IEEEbiographynophoto} 

\begin{IEEEbiographynophoto}
{Iwan M. Duursma} received his PhD in Mathematics from the University
of Eindhoven in 1993. After positions with CNRS IML Luminy, University
of Puerto Rico, Bell-Labs, AT\&T Research, and University of Limoges,
he is currently a professor at the University of Illinois at Urbana-Champaign.
\end{IEEEbiographynophoto} 

\begin{IEEEbiographynophoto}
{Han Mao Kiah} received his Ph.D. degree in mathematics from Nanyang Technological University (NTU), Singapore, in 2014. From 2014 to 2015 he was a Postdoctoral Research Associate at the Coordinated Science Laboratory, University of Illinois at Urbana-Champaign. Currently he is a Lecturer at the School of Physical and Mathematical Sciences, NTU, Singapore. His research interests include combinatorial design theory, coding theory, and enumerative combinatorics.
\end{IEEEbiographynophoto} 

\begin{IEEEbiographynophoto}
{Olgica Milenkovic} is a professor of Electrical and Computer Engineering at the University of Illinois, Urbana-Champaign (UIUC), and Research Professor at the Coordinated Science Laboratory. She obtained her Masters Degree in Mathematics in 2001 and PhD in Electrical Engineering in 2002, both from the University of Michigan, Ann Arbor. Prof. Milenkovic heads a group focused on addressing unique interdisciplinary research challenges spanning the areas of algorithm design and computing, bioinformatics, coding theory, machine learning and signal processing. Her scholarly contributions have been recognized by multiple awards, including the NSF Faculty Early Career Development (CAREER) Award, the DARPA Young Faculty Award, the Dean’s Excellence in Research Award, and several best paper awards. In 2013, she was elected a UIUC Center for Advanced Study Associate and Willett Scholar while in 2015 she was elected a Distinguished Lecturer of the Information Theory Society. In 2018 she became an IEEE Fellow. She has served as Associate Editor of the IEEE Transactions of Communications, the IEEE Transactions on Signal Processing, the IEEE Transactions on Information Theory and the IEEE Transactions on Molecular, Biological and Multi-Scale Communications. In 2009, she was the Guest Editor in Chief of a special issue of the IEEE Transactions on Information Theory on Molecular Biology and Neuroscience. 
\end{IEEEbiographynophoto}
\end{document}